\documentclass[12pt, draftclsnofoot, onecolumn]{IEEEtran}
\ifCLASSINFOpdf
\else
\fi
\usepackage{cite}
\usepackage{subfigure}
\usepackage{graphicx}
\usepackage{epstopdf}
\usepackage{extpfeil}
\usepackage{color}
\usepackage{bm}
\usepackage{algorithm}
\usepackage{algpseudocode}

\def\BibTeX{{\rm B\kern-.05em{\sc i\kern-.025em b}\kern-.08em
    T\kern-.1667em\lower.7ex\hbox{E}\kern-.125emX}}

\newtheorem{lem}{Lemma}
\newtheorem{thm}{Theorem}
\newtheorem{rmk}{Remark}

\newtheorem{pro}{Proposition}
\newenvironment{proof}{{\indent \indent \it Proof:\quad}}{\hfill $\blacksquare$\par}

\allowdisplaybreaks[4]

\begin{document}
%
\title{RIS-Assisted Quasi-Static Broad Coverage for Wideband mmWave Massive MIMO Systems}

\author{
Muxin~He, 
Jindan~Xu,
Wei~Xu,~\IEEEmembership{Senior Member,~IEEE}, 
Hong~Shen, 
Ning~Wang, 
and Chunming~Zhao
\thanks{M. He, J. Xu, W. Xu, H. Shen, and C. Zhao are with the National Mobile Communications Research Laboratory, Southeast University, Nanjing 210096, China (e-mail: mxhe@outlook.com; jdxu@seu.edu.cn;   wxu@seu.edu.cn; shhseu@seu.edu.cn;  cmzhao@seu.edu.cn). 
W. Xu is also with Henan Joint International Research Laboratory of Intelligent Networking and Data Analysis, Zhengzhou University, Zhengzhou 450001, China. C. Zhao is also with the Purple Mountain Laboratories, Nanjing, China. 

N. Wang is with the School of Information Engineering, Zhengzhou University, Zhengzhou 450001, China (ienwang@zzu.edu.cn).

This work has been submitted to the IEEE for possible publication.  Copyright may be transferred without notice, after which this version may no longer be accessible.
}
}

%



\maketitle
\vspace{-40pt}
\begin{abstract}
Reconfigurable intelligent surfaces (RISs) can establish favorable wireless environments to combat the severe attenuation and blockages in millimeter-wave (mmWave) bands.
However, to achieve the optimal enhancement of performance, the instantaneous channel state information (CSI) needs to be estimated at the cost of a large overhead that scales with the number of RIS elements and the number of users.
In this paper, we design a quasi-static broad coverage at the RIS with the reduced overhead based on the statistical CSI.
We propose a design framework to synthesize the power pattern reflected by the RIS that meets the customized requirements of broad coverage.
For the communication of broadcast channels, we generalize the broad coverage of the single transmit stream to the scenario of multiple streams. 
Moreover, we employ the quasi-static broad coverage for a multiuser orthogonal frequency division multiplexing access (OFDMA) system, and derive the analytical expression of the downlink rate, which is proved to increase logarithmically with the power gain reflected by the RIS.
By taking into account the overhead of channel estimation, the proposed quasi-static broad coverage even outperforms the design method that optimizes the RIS phases using the instantaneous CSI. Numerical simulations are conducted to verify these observations.
\end{abstract}
\vspace{-10pt}

%
\IEEEpeerreviewmaketitle

\section{Introduction}
%
%
%
%


\IEEEPARstart{R}{econfigurable} intelligent surfaces (RISs) are an innovative technology for implementing both spectral- and energy-efficient wireless networks beyond 5G \cite{Towards2019,Renzo2020Smart2,SH2019,ZSQ2020,SHENHONG2020,Zhaohui2021Beamforming,Shen2021}. It is well compatible with existing multi-antenna systems by mounting additional planar surfaces on the exterior walls of buildings. Each reflecting element independently changes the phase of the impinging signal such that an RIS is able to control the wireless environment. Empowered by this ability, RISs can improve the received signal power and can create optimized alternative propagation paths to bypass obstacles.

In millimeter-wave (mmWave) band, the wireless channel is vulnerable to the attenuation and blockages, which seriously affects the quality of service. RISs are, therefore, a promising technology for making mmWave communications more reliable \cite{Towards2019,di2019smart,Renzo2020Reconfigurable,Renzo2020Smart2}.
In \cite{Yaqiong2021}, the authors studied the cooperative reflection design for multiple RISs under the mmWave channels by taking the timing synchronization errors into account.
The authors of \cite{Perovi2020} designed the RIS to improve the channel capacity of a mmWave indoor environment without any line-of-sight (LoS) path.
In \cite{Wang2021}, the authors maximized the spectral efficiency of an RIS-assisted mmWave  system by exploiting the inherence structure of the cascaded mmWave channel.
The authors of \cite{Wenhui2021} estimated the cascaded channel for an RIS-assisted mmWave system with the quantized beamforming at the receiver.
In \cite{Pradhan2020}, the authors jointly designed the hybrid precoding at the base station (BS) and the phase shifts of the RIS in a multiuser mmWave system.

Regarding mmWave systems with wide band, the orthogonal frequency division multiplexing (OFDM) is commonly applied to combat the frequency-selective fading. However, the RIS design method proposed for the single-carrier system cannot be directly employed to the OFDM based multi-carrier system since an RIS cannot provide independent phase control for each subcarrier.
In \cite{OFDM9039554}, the authors maximized the achievable rate of an OFDM system by jointly optimizing the power allocation and the phase shifts of the RIS.
The authors of \cite{zhang2019} optimized the channel capacity of RIS-assisted OFDM systems by jointly designing the transmit covariance matrix and the phase shifts of the RIS.
In \cite{OFDM8937491}, the authors improved the performance of an OFDM system with a low-complexity method that matches the phase shifts of the RIS with the phase of the strongest channel path.
Considering the practical RIS model with dual phase- and amplitude- squint effect, the authors of \cite{Hongyu2021} studied the sum-rate maximization problem for a multiuser multi-antenna OFDM system with continuous and discrete phase shifts at the RIS.

Furthermore, the reflection design of RISs relies on the channel state information (CSI) among the BS, the RIS, and the user equipment (UE), which is difficult to acquire due to the nearly-passive hardware and the limited sensing capability of RISs \cite{Towards2019,Renzo2020Smart2}.
Even worse, the overhead of channel estimation is huge since a large number of reflecting elements are equipped at the RIS in order to compensate for the path losses of the BS-RIS channel and the RIS-UE channel \cite{Tang2021}.
Accordingly, the authors of \cite{OFDM9039554} and \cite{OFDM8937491} utilized a grouping method to reduce the training overhead for RIS-assisted OFDM systems, where the adjacent RIS reflecting elements were grouped to share a common reflection coefficient.
However, the channel estimation schemes proposed in \cite{OFDM9039554} and \cite{OFDM8937491} cannot be efficiently applied to the scenario of multiple users since the UE-by-UE successive channel estimation leads to a large overhead that scales with the number of users.
In \cite{Zheng2020}, the authors proposed two channel estimation schemes for the RIS-assisted multiuser orthogonal frequency division multiplexing access (OFDMA) system, where the maximum number of supported users was proved to be limited, and an increment of users was obtained at the expense of higher complexity and degraded performance.
It was shown in \cite{OFDM9039554}, \cite{OFDM8937491}, and \cite{Zheng2020} that the large overhead of training and feedback poses a challenge for RIS-assisted OFDM systems.

Since the acquisition of full CSI is challenging with low-cost RIS circuits, some studies, e.g., \cite{Youli2020,Nadeem2020,Gao2021}, investigated RIS-assisted communications with partial CSI. Authors of \cite{Youli2020} discussed a tradeoff between the energy efficiency and the spectral efficiency of an RIS-assisted multiuser multi-antenna system by using the partial CSI between the RIS and the users.
In \cite{Nadeem2020}, with the CSI of signal-to-noise-ratio (SNR), the authors proposed an RIS-assisted opportunistic beamforming for broadcast channels.
In a multi-antenna system assisted by multiple distributed RISs, the authors of \cite{Gao2021} proposed a low-complexity RIS design by exploiting the statistical correlation information of channels.

In addition, in the communication of broadcast channels, not all the users' CSI are available.
Therefore, the transmit signals are commonly precoded into a broad beam to cover the users at different locations, where the flat-top pattern with small power fluctuations is preferred.
In conventional massive multiple-input multiple-output (MIMO) systems, the authors of \cite{Meng2016} utilized the Zadoff-Chu sequence to design a channel-independent omnidirectional precoding that maintained the equal average radiation power in each spatial direction.
In \cite{Meng2018}, the authors proposed the channel-independent omnidirectional space-time block coding to satisfy a more strict constant power constraint at any instant time.
To obtain a flexible sector size of coverage, the authors of \cite{guowr2019} utilized the manifold optimization to synthesize the target flat-top pattern with negligible power fluctuations, which needed no CSI of users.

Considering the above difficulties in acquiring the CSI for the RIS design, the RIS-assisted broad coverage can be developed by broadening the power pattern reflected by the RIS.
In \cite{Jamali2021}, the authors employed the quadratic gradient phase shifts at the RIS instead of the linear gradient phase shifts to broaden the power pattern reflected by the RIS.
However, the control of the power fluctuations within the broad beam was not mentioned.
Other methods of beam broadening include the deactivation based technique and sub-array based technique \cite{Xiao2016}.
The deactivation based technique generates a wider beam by turning off part of the antennas in an array, which, however, limits the maximal total transmit power when the number of active antennas is small \cite{Xiao2016}.
The sub-array based technique points the beam of each sub-array towards a separated angle, and combines the multiple beams of the sub-arrays into a broad beam.
In \cite{Hai2021}, the authors proposed a sub-array based beam broadening and flattening technique to the design of RIS phase shifts, which achieved better performance than the deactivation based method.
However, the power pattern generated by \cite{Hai2021} cannot admit a continuous adjustment of the beamwidth due to the limited spatial resolution of the sub-arrays, especially for a relatively small array.
In \cite{He2021}, inspired by \cite{guowr2019}, we adopted the manifold method to synthesize a customized flat-top power pattern with an arbitrarily defined beamwidth at the RIS.
Since the non-convex constant modulus constraint on the phase shift is a challenge of the RIS design, we can regard this constraint as a manifold to facilitate the optimization \cite{Yu2019}.

Note that the above beam broadening techniques relies on the single incident angle of arrival at the RIS, e.g., the LoS channel path. However, in the RIS-assisted communication of broadcast channels, the transmission design using the LoS BS-RIS channel results in an equivalent rank-1 cascaded BS-RIS-UE channel, which hinders the transmission of multiple streams.
Moreover, when the multi-stream transmission is supported by multiple channel paths between the BS and the RIS, the design method based on the LoS BS-RIS channel, e.g., as proposed in \cite{Hai2021} and \cite{He2021}, cannot ensure a small power fluctuation within the coverage due to the interference from the non-LoS (NLoS) BS-RIS channel paths.

To address the aforementioned issues, we design the quasi-static broad coverage at the RIS by using the statistical CSI to reduce the overhead.
Moreover, we utilize multiple mmWave BS-RIS channel paths in the pattern synthesis such that the rank of the cascaded channel matrix can support the transmission of multiple streams.
Under the proposed quasi-static broad coverage, we analyze the system performance in both the broadcast channel and the OFDMA channel.
The main contributions of our work are summarized as follows.

1) We propose a design framework to synthesize the flat-top power pattern reflected by the RIS in a mmWave massive MIMO-OFDM system, where the precoder at the BS and the phase shifts of the RIS are jointly designed.
We regard the constant modulus constraint of the RIS phase shifts as a Riemannian manifold such that we can perform the conjugate gradient (CG) method to minimize the distance between the target pattern and the power pattern reflected by the RIS.
The proposed design framework is general since it can be applied to the synthesis of a customized target pattern and to the scenario of multipath channel models. 
With the proposed design framework, we can provide a quasi-static broad coverage for the users by exploiting only the statistical CSI of mmWave channel paths, which saves the overhead of channel estimation and feedback.

2) In the communication of broadcast channels, 
we generalize RIS-assisted broad coverage of the single transmit stream to the scenario of multiple transmit streams. 
When the BS-RIS channel is simplified as an LoS channel to formulate the reflected flat-top pattern, the equivalent cascaded BS-RIS-UE channel is rank-1 and thus supports only a single transmit stream. 
In this paper, we employ a more general multipath mmWave model for the BS-RIS channel, which poses new challenge to the pattern synthesis. 
Therefore, multiple transmit streams can be supported by the equivalent high-rank cascaded BS-RIS-UE channel, which, to the best of our knowledge, has not been provided in the literature.
Moreover, it is proved that, within the broad coverage, each user receives a constant average received power.

3) In order to provide reliable communications for multiuser OFDMA systems, we design a quasi-static flat-top power pattern at the RIS to cover the users instead of optimizing the RIS phase shifts with the instantaneous CSI.
By using the dominant LoS channel, the proposed quasi-static broad coverage needs no more channel estimation for controlling the RIS, whose overhead approaches that of the conventional OFDMA system.
Therefore, by taking the overhead into consideration, the proposed broad coverage even outperforms the RIS design method that needs the estimation of the instantaneous CSI.
Furthermore, we derive an analytical expression of the downlink rate of the considered OFDMA system, and find that the rate increases logarithmically with the power gain reflected by the RIS.

The rest of this paper is organized as follows. Firstly, the considered system model is introduced in Section \ref{sec:systemmodel}. The manifold optimization is presented in Section \ref{sec:manifold}. In Section \ref{sec:applicationandperformance}, we discuss the applications of broad coverage and analyze the system performance. Finally, simulations are conducted in Section \ref{sec:numerical}, and conclusions are drawn in Section \ref{sec:conclusion}.

\emph{Notation:} Lowercase and uppercase boldface letters represent column vectors and matrices, respectively. ${\bf A}^*, {\bf A}^T$, and ${\bf A}^H$ denote the conjugate, the transpose, and the Hermitian transpose of matrix ${\bf A}$, respectively. 
${\rm diag}({\bf a})$ is the diagonal matrix with vector ${\bf a}$ as its diagonal entries, and ${\rm diag}({\bf A})$ takes only the diagonal entries of matrix ${\bf A}$.
${\rm vec}\{\cdot\}$ denotes the vectorization operator that stacks the columns of the input matrix into a long column.
$\|\cdot\|_2$ and $\|\cdot\|_F$ denote the 2-norm and the Frobenius norm, respectively. ${\rm tr}({\bf A})$ denotes the trace of matrix ${\bf A}$. $\Re[\cdot]$ denotes the real part of the input. $\circ$ denotes the Hadamard product. 
$\jmath$ denotes the imaginary unit. 
${\mathbb C}^{m\times n}$ stands for the set of ${m\times n}$ complex matrices. 
${\mathbb Z}$, ${\mathbb R}^+$, and $\emptyset$ denote the integer set, the set of positive real number, and the empty set, respectively.

\section{System Model}\label{sec:systemmodel}

\subsection{Signal Model}
\begin{figure}[!t]
\centering
\vspace{-5pt}
\includegraphics[width=13pc]{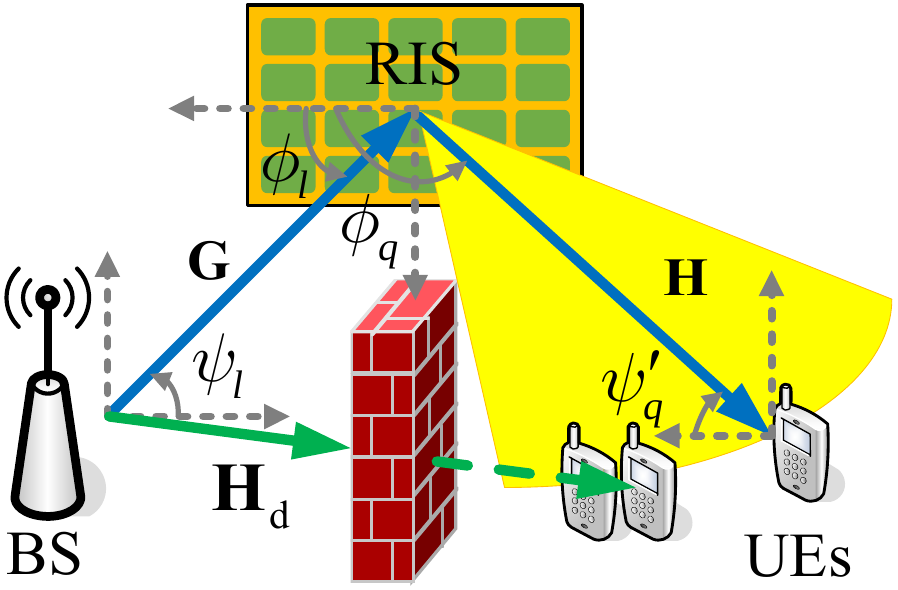}
\caption{RIS-assisted quasi-static broad coverage in a mmWave massive MIMO system.}
\label{fig:overview}
\vspace{-10pt}
\end{figure}

As shown in Fig. \ref{fig:overview}, we consider the RIS-assisted broad coverage in a mmWave massive MIMO system, where several users with $N_{\rm UE}$ antennas are served by a BS with $N_{\rm BS}$ antennas and assisted by an RIS with $M$ reflecting elements. In the considered system, the BS fails to provide reliable communications for the users due to the obstacle between them, and an RIS is deployed around the obstacle to provide a broad coverage for the blocked users.
For an OFDM-based downlink transmission, the received signal ${\bf r}[k]\in{\mathbb C}^{N_{\rm UE} \times 1}$, at the $k$-th subcarrier, of one given user is
\begin{align}\label{eq:signal}
{\bf r}[k]=\sqrt{p} \left(\sqrt{\beta_1 \beta_2}{\bf H}[k]{\bf \Theta}{\bf G}[k]+\sqrt{\beta}{\bf H}_{\rm d}[k]\right){\bf W}[k]{\bf s}[k]+{\bf z}[k],
\end{align}
\noindent where $p$ is the transmit power of the BS, $\beta_1$, $\beta_2$, and $\beta$ are the large-scale fading factors of the BS-RIS channnel, the RIS-UE channel, and the BS-UE direct channel, respectively, ${\bf H}[k] \in {\mathbb C}^{N_{\rm UE} \times M}$ is the channel from the RIS to the user,
${\bf \Theta} = {\rm diag}\left(\theta_{0} ,\theta_{1} ,..., \theta_{M - 1}\right)$ is the phase matrix at the RIS,
${\bf G}[k]\in{\mathbb C}^{M \times N_{\rm BS}}$ is the channel from the BS to the RIS, ${\bf H}_{\rm d}[k]\in{\mathbb C}^{N_{\rm UE} \times N_{\rm BS}}$ is the direct channel from the BS to the user,
${\bf W}[k] \in {\mathbb C}^{N_{\rm BS} \times N_{\rm d}}$ is the precoding matrix of $N_{\rm d}$ transmit streams at the BS, ${\bf s}[k]\in{\mathbb C}^{N_{\rm d} \times 1}$ is the transmit signal with ${\mathbb E}\{{\bf s}[k]{\bf s}^H[k]\}={\bf I}_{N_{\rm d}}$, and ${\bf z}[k] \in {\mathbb C}^{N_{\rm UE} \times 1}$ is the additive white Gaussian noise (AWGN) following ${\mathcal{CN}}({\bf 0},\sigma_z^2{\bf I}_{N_{\rm UE}})$.
Note that the precoding matrices at $N_{\rm c}$ subcarriers are normalized to $\sum_{k=0}^{N_{c}-1}\left\|{\bf W}[k]\right\|_F^2=N_{\rm c}$.

\subsection{Channel Model}
As for the time-domain multipath mmWave channels ${\bf G}[n]$, ${\bf H}[n]$, and ${\bf H}_{\rm d}[n]$,
we adopt the geometric channel model as follows\cite{Gao2016,Venugopal2017}
\begin{align}
{\bf G}[n]&\triangleq \sqrt{N_{\rm BS}M}\sum_{l=0}^{L-1}\alpha_{l}{\bf a}_{\bf G}(\phi_{l}){\bf b}_{\bf G}^H(\psi_{l})g(nT-\tau_l),\label{eq:G}\\
{\bf H}[n]&\triangleq \sqrt{N_{\rm UE}M}\sum_{q=0}^{Q-1}\alpha_{q}{\bf b}_{\bf H}(\psi'_{q}){\bf a}_{\bf H}^H(\phi_{q})g(nT-\tau_{q}),\label{eq:H}\\
{\bf H}_{\rm d}[n]&\triangleq \sqrt{N_{\rm BS}N_{\rm UE}}\sum_{q'=0}^{Q'-1}\alpha_{q'}{\bf b}_{\bf H}(\psi'_{q'}){\bf b}_{\bf G}^H(\psi_{q'})g(nT-\tau_{q'}),\label{eq:Hd}
\end{align}
\noindent where $n$ is the sampling index, $L$, $Q$, and $Q'$ are the numbers of paths of the BS-RIS channel, the RIS-UE channel, and the BS-UE direct channel, respectively,
$\alpha_l$, $\alpha_q$, and $\alpha_{q'}$ are the complex gains of the $l$-th path of ${\bf G}[n]$, the $q$-th path of ${\bf H}[n]$, and the $q'$-th path of ${\bf H}_{\rm d}[n]$, respectively, 
$\phi_l$ and $\phi_q$ are the angles of arrival (AoAs) of the $l$-th path of ${\bf G}[n]$ and the angles of departure (AoDs) of the $q$-th path of ${\bf H}[n]$, respectively, 
$\psi_l$ and $\psi_{q'}$ are the AoDs of the $l$-th path of ${\bf G}[n]$ and the $q'$-th path of ${\bf H}_{\rm d}[n]$, respectively, 
$\psi'_{q}$ and $\psi'_{q'}$ are the AoAs of the $q$-th path of ${\bf H}[n]$ and the $q'$-th path of ${\bf H}_{\rm d}[n]$, respectively,  
${\bf a}_{\bf G}\in{\mathbb C}^{M \times 1}$ and ${\bf a}_{\bf H}\in{\mathbb C}^{M \times 1}$ are the array response vectors (ARVs) at the RIS that correspond to the signal arrival and the signal departure, respectively,
${\bf b}_{\bf G}\in{\mathbb C}^{N_{\rm BS} \times 1}$ and ${\bf b}_{\bf H}\in{\mathbb C}^{N_{\rm UE} \times 1}$ are the ARVs at the BS and the user, respectively,
$T$ is the sampling period, $\tau_l$, $\tau_{q}$, and $\tau_{q'}$ are the transmission delays of the $l$-th path of ${\bf G}[n]$, the $q$-th path of ${\bf H}[n]$, and the $q'$-th path of ${\bf H}_{\rm d}[n]$, respectively, 
and $g(t)$ is the rectangular pulse shaping filter, i.e., $g(t)=1$ only for $-T\leq t<0$. Moreover, we assume that the channel paths undergo uncorrelated fading.
Since a planar array can be treated as two decoupled linear arrays in the pattern synthesis \cite{Hai2021}, in this paper, we use the model of uniform linear array at the BS, the RIS, and the UE for simplification and define the ARVs as follows
\begin{align}
{\bf a}_{\bf G}(\phi_l)&\triangleq\frac{1}{\sqrt{M}}\left[1, e^{-\jmath \frac{2\pi}{\lambda}\rho \cos \phi_l}, ..., e^{-\jmath \frac{2\pi}{\lambda}(M-1)\rho \cos \phi_l}\right]^T,\label{eq:aG}\\
{\bf b}_{\bf G}(\psi_l)&\triangleq\frac{1}{\sqrt{N_{\rm BS}}}\left[1, e^{-\jmath \frac{2\pi}{\lambda}\rho \sin \psi_l}, ..., e^{-\jmath \frac{2\pi}{\lambda}(N_{\rm BS}-1)\rho \sin \psi_l}\right]^T, \label{eq:bG}\\
{\bf a}_{\bf H}(\phi_q)&\triangleq\frac{1}{\sqrt{M}}\left[1, e^{\jmath \frac{2\pi}{\lambda}\rho \cos \phi_q}, ..., e^{\jmath \frac{2\pi}{\lambda}(M-1)\rho \cos \phi_q}\right]^T,\label{eq:aH}\\
{\bf b}_{\bf H}(\psi'_q)&\triangleq\frac{1}{\sqrt{N_{\rm UE}}}\left[1, e^{-\jmath \frac{2\pi}{\lambda}\rho \sin \psi'_q}, ..., e^{-\jmath \frac{2\pi}{\lambda}(N_{\rm UE}-1)\rho \sin \psi'_q}\right]^T,\label{eq:bH}
\end{align}
\noindent where $\lambda$ is the wavelength and $\rho=\lambda/2$ is the distance between adjacent elements. 

As for the frequency-domain channels, we apply an $N_{\rm c}$-dimensional discrete Fourier transform (DFT) to (\ref{eq:G})-(\ref{eq:Hd}), and obtain
\begin{align}
{\bf G}[k]&=\sum_{n=0}^{N_{\rm c}-1}{\bf G}[n]e^{-\jmath \frac{2\pi kn}{N_{\rm c}}}
\triangleq\sqrt{N_{\rm BS}M}\sum_{l=0}^{L-1}\delta_l[k]{\bf a}_{\bf G}(\phi_{l}){\bf b}_{\bf G}^H(\psi_{l}),\label{eq:Gfreq}\\
{\bf H}[k]&=\sum_{n=0}^{N_{\rm c}-1}{\bf H}[n]e^{-\jmath \frac{2\pi kn}{N_{\rm c}}}
\triangleq\sqrt{N_{\rm UE}M}\sum_{q=0}^{Q-1}\delta_q[k]{\bf b}_{\bf H}(\psi'_{q}){\bf a}_{\bf H}^H(\phi_{q}),\label{eq:Hfreq}\\
{\bf H}_{\rm d}[k]&=\sum_{n=0}^{N_{\rm c}-1}{\bf H}_{\rm d}[n]e^{-\jmath \frac{2\pi kn}{N_{\rm c}}}
\triangleq\sqrt{N_{\rm BS}N_{\rm UE}}\sum_{q'=0}^{Q'-1}\delta_{q'}[k]{\bf b}_{\bf H}(\psi'_{q'}){\bf b}_{\bf G}^H(\psi_{q'}),\label{eq:Hdfreq}
\end{align}
\noindent where $\delta_l[k]\triangleq\sum_{n=0}^{N_{\rm c}-1}\alpha_{l}g(nT-\tau_l)e^{-\jmath \frac{2\pi kn}{N_{\rm c}}}=\alpha_{l}e^{-\jmath \frac{2\pi kn_l}{N_{\rm c}}}$,
$\delta_q[k]\triangleq\sum_{n=0}^{N_{\rm c}-1}\alpha_{q}g(nT-\tau_{q})e^{-\jmath \frac{2\pi kn}{N_{\rm c}}}=\alpha_{q}e^{-\jmath \frac{2\pi kn_q}{N_{\rm c}}}$, and $\delta_{q'}[k]\triangleq\sum_{n=0}^{N_{\rm c}-1}\alpha_{q'}g(nT-\tau_{q'})e^{-\jmath \frac{2\pi kn}{N_{\rm c}}}=\alpha_{q'}e^{-\jmath \frac{2\pi kn_{q'}}{N_{\rm c}}}$
with the indices $n_l$, $n_q$, and $n_{q'}$ satisfying $n_lT<\tau_l\leq(n_l+1)T$,  $n_qT<\tau_q\leq(n_q+1)T$, and $n_{q'}T<\tau_{q'}\leq(n_{q'}+1)T$, respectively.

Moreover, the channel model in (\ref{eq:Gfreq}) can be written in a compact form as follows
\begin{align}\label{eq:Gcom}
{\bf G}[k]
={\bf A}_{\bf G}{\bf \Delta}_{\bf G}[k]{\bf\bar \Delta}_{\bf G}{\bf B}_{\bf G}^H,
\end{align}
\noindent where ${\bf A}_{\bf G}\triangleq[{\bf a}_{\bf G}(\phi_0),{\bf a}_{\bf G}(\phi_1),...,{\bf a}_{\bf G}(\phi_{L-1})]$,
${\bf \Delta}_{\bf G}[k]\triangleq{\rm diag}\left(e^{-\jmath \frac{2\pi kn_0}{N_{\rm c}}},e^{-\jmath \frac{2\pi kn_1}{N_{\rm c}}},...,e^{-\jmath \frac{2\pi kn_{L-1}}{N_{\rm c}}}\right)$,
${\bf\bar \Delta}_{\bf G}\triangleq{\rm diag}(\alpha_0,\alpha_1,...,\alpha_{L-1})$,
and ${\bf B}_{\bf G}\triangleq[{\bf b}_{\bf G}(\psi_0),{\bf b}_{\bf G}(\psi_1),...,{\bf b}_{\bf G}(\psi_{L-1})]$.

\subsection{Power Pattern of the RIS}

For the transmission of $N_{\rm d}$ streams, we can treat the power pattern as the summation of $N_{\rm d}$ patterns, each of which is generated by one column of ${\bf W}$ \cite{guowr2019}.
Therefore, we obtain the power pattern reflected by the RIS, towards the angle $\phi$ and at the $k$-th subcarrier, as follows
\begin{align}\label{eq:pattern}
y(\phi,k)&=M\left\|{\bf a}_{\bf H}^H(\phi){\bf \Theta}{\bf G}[k]{\bf W}[k]\right\|_2^2,
\end{align}
\noindent where ${\bf a}_{\bf H}(\phi)$, defined in (\ref{eq:aH}), represents the steering vector towards $\phi$. We omit the transmit power $p$ and the large-scale fading factors in (\ref{eq:pattern}) for simplification.
Note that the power pattern $y(\phi,k)$ defined in (\ref{eq:pattern}) differs among the subcarriers since it depends on the BS-RIS channel ${\bf G}[k]$ at each subcarrier.

However, the acquirement of instantaneous CSI in the RIS-assisted OFDM system requires a high overhead of training and feedback, therefore, it is viable to synthesize the pattern by using the statistical CSI between the BS and the RIS. 
In this way, we take the expectation of (\ref{eq:pattern}) and obtain the average power pattern $y(\phi)$ as follows
\begin{align}\label{eq:avepattern}
y(\phi)
&={M}{\mathbb E}\left\{\|{\bf a}_{\bf H}^H(\phi){\bf \Theta}{\bf G}[k]{\bf W}[k]\|_2^2\right\}\nonumber\\
&=M^2 N_{\rm BS}{\bf a}_{\bf H}^H(\phi){\bf \Theta}{\bf A}_{\bf G}{\bf \Delta}_{\bf G}[k]{\mathbb E}\left\{{\bf\bar \Delta}_{\bf G}{\bf B}_{\bf G}^H{\bf W}[k]
{\bf W}[k]^H{\bf B}_{\bf G}{\bf\bar \Delta}_{\bf G}^H\right\}{\bf \Delta}_{\bf G}^H[k]{\bf A}_{\bf G}^H{\bf \Theta}^H{\bf a}_{\bf H}(\phi)\nonumber\\
&=M^2 N_{\rm BS}{\bf a}_{\bf H}^H(\phi){\bf \Theta}{\bf A}_{\bf G}{\bf I}\circ\left({\bf\Lambda}{\bf B}_{\bf G}^H{\bf W}
{\bf W}^H{\bf B}_{\bf G}\right){\bf A}_{\bf G}^H{\bf \Theta}^H{\bf a}_{\bf H}(\phi),
\end{align}
\noindent where ${\bf\Lambda}\triangleq{\mathbb E}\left\{{\bf\bar \Delta}_{\bf G}{\bf\bar \Delta}_{\bf G}^H\right\}={\rm diag}\left({\mathbb E}\{|\alpha_0|^2\},{\mathbb E}\{|\alpha_1|^2\},...,{\mathbb E}\{|\alpha_{L-1}|^2\}\right)$.
Note that the average power pattern derived in (\ref{eq:avepattern}) is independent of the index of subcarrier $k$, therefore, the design of the precoder ${\bf W}$ and the RIS phase matrix ${\bf \Theta}$ is frequency-independent. 
Moreover, the RIS provides a quasi-static coverage for the users since the average power pattern derived in (\ref{eq:avepattern}) relies on the statistical CSI of the channel paths, i.e., ${\bf A}_{\bf G}$, ${\bf B}_{\bf G}$, and ${\bf\Lambda}$.

\begin{rmk}\label{rmk:LOS}
	When the power pattern is synthesized by using the LoS channel between the BS and the RIS, as previously discussed in \cite{He2021}, we find that the power pattern in (\ref{eq:pattern}) and the corresponding pattern synthesis are essentially frequency-independent.
\end{rmk}

To facilitate the pattern synthesis, we oversample the continuous pattern $y(\phi)$ derived in (\ref{eq:avepattern}) at the discrete angles of ${\bar \phi_j}=\frac{\pi}{\kappa M}j$, $j=0,1,...,\kappa M-1$, where $\kappa>1$. Then, we obtain the discrete power pattern ${\bf y}$ as follows \cite{guowr2019}
\begin{align}\label{eq:avegdis}
{\bf  y}=M^2 N_{\rm BS}{\rm diag}\left({\bf\widetilde A}{\bf \Theta}{\bf A}_{\bf G}{\bf I}\circ\left({\bf\Lambda}{\bf B}_{\bf G}^H{\bf W}
{\bf W}^H{\bf B}_{\bf G}\right){\bf A}_{\bf G}^H{\bf \Theta}^H{\bf\widetilde A}^H\right),
\end{align}
\noindent where
\begin{align}\label{eq:T}
{\bf\widetilde A}\triangleq\left[{\bf a}_{\bf H}({\bar \phi_0}),{\bf a}_{\bf H}({\bar \phi_1}),...,{\bf a}_{\bf H}({\bar \phi_{\kappa M-1}})\right]^H.
\end{align}

In order to meet the requirement of the broad coverage, the reflected power at each angle in (\ref{eq:avegdis}) needs to be synthesized into the corresponding desired value, which defines a target power pattern.
To this end, we formulate the desired broad coverage by defining the target pattern in the next subsection.

\subsection{Target Pattern}

To provide a reliable broad coverage for the blocked users, the target pattern is defined as a flat-top pattern as follows
\begin{align}\label{eq:targetf}
f(\phi)\triangleq
\begin{cases}
f_{\rm M},&\ |\phi-\phi_{\rm c}|\leq \phi_{0.5}(1-\varepsilon)\\
f_{\rm S},&\ |\phi-\phi_{\rm c}|>\phi_{0.5}(1+\varepsilon)\\
\mu\left(\frac{\pi\left[|\phi-\phi_{\rm c}|-\phi_{0.5}(1-\varepsilon)\right]}{2\varepsilon \phi_{0.5}}\right),&\ \text{otherwise,}
\end{cases}
\end{align}
\noindent where $f_{\rm M}$ denotes the power of the flat-top beam, $f_{\rm S}$ denotes the power of the side lobe, 
$\phi_{\rm c}$ denotes the angle that the center of the flat-top beam points towards, 
$\phi_{0.5}$ denotes half of the beamwidth of the flat-top beam,  $\mu(x)\triangleq\frac{f_{\rm M}+f_{\rm S}}{2}+\frac{f_{\rm M}-f_{\rm S}}{2}\cos x$, and $\varepsilon$ is the roll-off factor. Note that the target pattern in (\ref{eq:targetf}) provides a broad coverage with the constant power gain over the angle range of $[\phi_{\rm c}-\phi_{0.5},\phi_{\rm c}+\phi_{0.5}]$. 
Then, the target pattern is discretized into a vector ${\bf  f}$ with the same oversampling factor $\kappa$ as defined in (\ref{eq:avegdis}).

Therefore, we can synthesize the power pattern in (\ref{eq:avegdis}) according to the target pattern defined in (\ref{eq:targetf}). In the next subsection, we formulate the optimization problem of pattern synthesis.

\subsection{Problem Formulation}\label{subsec:formulation}

To obtain the desired pattern, we synthesize the power pattern ${\bf  y}$ derived in (\ref{eq:avegdis}) by jointly optimizing the precoder ${\bf W}$ and the RIS phase matrix ${\bf \Theta}$.
To start with, we define the cost function of the optimization problem as the distance between the power pattern ${\bf  y}$ and the target pattern ${\bf f}$.
In \cite{guowr2019} and \cite{He2021}, the cosine distance was used to depict the difference between ${\bf  y}$ and ${\bf  f}$, where the magnitude information was not considered \cite{tan2015}.
When using the LoS BS-RIS channel model for the design, we can decouple the BS precoder from the power pattern and handle the magnitude of the pattern independently \cite{He2021}.
However, when considering the multipath channel model between the BS and the RIS, the precoder ${\bf W}$ has an impact on both the magnitude and the shape of the power pattern, and cannot be decoupled from the power pattern as shown in (\ref{eq:avegdis}).
The lack of magnitude information can lead to a degradation of the power gain when pursuing a well-shaped pattern.
Therefore, we use a weighted Euclidean distance in the following to take the magnitude of power pattern into consideration.

As defined in (\ref{eq:targetf}), the target power pattern is divided into three regions: the flat-top region $\mathcal{A}_1=\{\phi\big|\ |\phi-\phi_{\rm c}|\leq \phi_{0.5}(1-\varepsilon)\}$, the sidelobe region $\mathcal{A}_2=\{\phi\big|\ |\phi-\phi_{\rm c}|>\phi_{0.5}(1+\varepsilon)\}$, and the roll-off transition region $\mathcal{A}_3=\{\phi\big|\ \phi_{0.5}(1-\varepsilon)<|\phi-\phi_{\rm c}|\leq \phi_{0.5}(1+\varepsilon)\}$. 
To obtain the desired broad coverage, we can adopt different weights, i.e., ${\bar \gamma}_{\imath}\in{\mathbb R}^+$, ${\imath}=1,2,3$, for the three regions according to their priorities. Therefore, we define the weighted Euclidean distance between ${\bf  y}$ and ${\bf  f}$ as follows
\begin{align}\label{D}
D({\bf  f},{\bf  y})\triangleq \left[\sum_{j=0}^{\kappa M-1}\gamma_j\left({\bf f}_j-{\bf  y}_j\right)^2\right]^{1/2},
\end{align}
\noindent where the index $j$ corresponds to the discrete angle ${\bar \phi}_j$, 
and the weight $\gamma_j$ is defined as follows
\begin{align}\label{eq:weight}
\gamma_j\triangleq
\begin{cases}
{\bar \gamma}_{\imath},&\ {\bar \phi}_j\in\mathcal{A}_{\imath}, \ {\imath}=1,3\\
{\bar \gamma}_2,&\ {\bar \phi}_j\in\mathcal{A}_2 \ \text{and}\  {\bf y}_j>{\bf f}_j\\
0,&\ {\bar \phi}_j\in\mathcal{A}_2\ \text{and}\  {\bf y}_j\leq{\bf  f}_j.
\end{cases}
\end{align}
\noindent Note that the side lobes that are lower than the target pattern, i.e., ${\bf y}_j\leq{\bf  f}_j$, are not included.
By defining the weight matrix as ${\bf \Gamma}\triangleq{\rm diag}\left(\sqrt{\gamma_0},\sqrt{\gamma_1},...,\sqrt{\gamma_{\kappa M-1}}\right)$, we rewrite the weighted Euclidean distance in (\ref{D}) by a compact form as follows
\begin{align}\label{eq:D2}
D({\bf  f},{\bf  y})=\left\|{\bf \Gamma}\left({\bf f}-{\bf  y}\right)\right\|_2.
\end{align}

Equivalently, we can formulate the optimization problem of pattern synthesis as
\begin{align}
{\rm P1}:&\ \mathop{\min}_{{\bf W},{\bf \Theta}}\ \
J({\bf W},{\bf \Theta})
=\left\|{\bf \Gamma}\left[{\bf f}-{\bf  y}({\bf W},{\bf \Theta})\right]\right\|_2^2\\
&\ \ \ {\rm s.t.}\ \ \left\|{\bf W}\right\|_F^2=1,\label{eq:opcon2}\\
&\ \ \ \ \ \ \ \ \ |\theta_m|=1,\ m=0,1,...,M-1,\label{eq:opcon3}
\end{align}
\noindent where (\ref{eq:opcon2}) and (\ref{eq:opcon3}) are the BS transmit power constraint and the RIS phase shift constraint, respectively. 
Note that the normalization of the precoder, i.e., $\sum_{k=0}^{N_{c}-1}\left\|{\bf W}[k]\right\|_F^2=N_{\rm c}$ defined in (\ref{eq:signal}), reduces to (\ref{eq:opcon2}) due to the frequency-independent precoder.

In problem P1, there exists two issues on the optimization.
Firstly, the variables ${\bf W}$ and ${\bf \Theta}$ are coupled in the cost function.
Secondly, the phase shifts of the RIS are constrained by a constant modulus, which makes the optimization intractable.
To address the above issues, we propose a design framework for solving problem P1 in the following section.

\section{Manifold Optimization For Pattern Synthesis}\label{sec:manifold}

\subsection{Problem Reformulation}\label{sec:reformulation}
As mentioned above, problem P1 is hard to tackle and needs to be reformulated.
Observing the constraint in (\ref{eq:opcon2}), we divide ${\bf W}$ by a normalization factor $\left\|{\bf W}\right\|_F$, and directly substitute the power constraint into the power pattern as follows
\begin{align}\label{eq:patternbar}
{\bf\bar  y}=\frac{M^2 N_{\rm BS}}{\|{\bf W}\|_F^2}{\rm diag}\left({\bf\widetilde A}{\bf \Theta}{\bf A}_{\bf G}{\bf I}\circ\left({\bf\Lambda}{\bf B}_{\bf G}^H{\bf W}
{\bf W}^H{\bf B}_{\bf G}\right){\bf A}_{\bf G}^H{\bf \Theta}^H{\bf\widetilde A}^H\right).
\end{align}
Therefore, problem P1 is reformulated as
\begin{align}
{\rm P2}:&\ \mathop{\min}_{{\bf W},{\bf \Theta}}\ \
J({\bf W},{\bf \Theta})
=\left\|{\bf \Gamma}\left[{\bf f}-{\bf \bar y}({\bf W},{\bf \Theta})\right]\right\|_2^2\label{eq:P2J}\\
&\ \ \ {\rm s.t.}\ \ |\theta_m|=1,\ m=0,1,...,M-1.\label{eq:opcon32}
\end{align}
It is observed that the variables ${\bf W}$ and ${\bf \Theta}$ are coupled in (\ref{eq:patternbar}) and (\ref{eq:P2J}). To handle this, we employ an alternating optimization method to obtain a locally optimal solution to problem P2, which iteratively minimizes the cost function with respect to ${\bf W}$ and ${\bf \Theta}$ while keeping the other one fixed. Therefore, we have the following subproblems with respect to ${\bf W}$ and ${\bf \Theta}$, respectively:
\begin{align}
{\rm P2.1}:&\ \mathop{\min}_{{\bf W}}\ \
J({\bf W})
=\left\|{\bf \Gamma}\left[{\bf f}-{\bf \bar y}({\bf W})\right]\right\|_2^2\label{eq:P22}\\
{\rm P2.2}:&\ \mathop{\min}_{{\bf \Theta}}\ \
J({\bf \Theta})
=\left\|{\bf \Gamma}\left[{\bf f}-{\bf \bar y}({\bf \Theta})\right]\right\|_2^2\label{eq:P21}\\
&\ \ \ {\rm s.t.}\ \ |\theta_m|=1,\ m=0,1,...,M-1.\label{eq:P21con}
\end{align}

Given the non-convex constant modulus constraint in (\ref{eq:P21con}), it is difficult to solve problem P2.2 by directly using conventional method of pattern synthesis over the Euclidean space.
To overcome this problem, we regard the constant modulus constraint as a Riemannian manifold such that we can search the solution over the manifold.
In the following subsections, we first optimize the precoder ${\bf W}$ in problem P2.1 by using the conventional CG method. Next, we propose a manifold based CG method for problem P2.2 to optimize the phase shifts ${\bf \Theta}$ of the RIS. Then, we summarize the alternating optimization for ${\bf W}$ and ${\bf \Theta}$. 

\subsection{Optimization for ${\bf W}$}

For problem P2.1 with a given ${\bf \Theta}$, we can employ the conventional CG method to search the optimal ${\bf W}$\cite{bertsekas1997}, where the Euclidean gradient of the cost function $J({\bf W})$ is derived in the following proposition.
\begin{pro}\label{pro:EucgradW}
	The Euclidean gradient of $J({\bf W})$ has the following form:
	\begin{align}\label{eq:gradW}
	\frac{\partial J}{\partial{\bf W}^*}
	=&\frac{2{\bf W}}{\|{\bf W}\|_F^2}
	\left[{\rm tr}\left({\bf \Gamma}^2{\bf f}{\bf\bar y}^H\right)
	-\left\|{\bf \Gamma}{\bf\bar y}\right\|_2^2\right]\nonumber\\
	&+\frac{2M^2N_{\rm BS}}{\|{\bf W}\|_F^2}
	{\bf B}_{\bf G}{\bf I}\circ\left[{\bf A}_{\bf G}^H{\bf \Theta}^H{\bf\widetilde A}^H{\bf \Gamma}^2{\rm diag}\left({\bf\bar y}-{\bf f}\right){\bf\widetilde A}{\bf \Theta}{\bf A}_{\bf G}\right]{\bf\Lambda}{\bf B}_{\bf G}^H{\bf W}.
	\end{align}
\end{pro}

\begin{proof}
	See Appendix \ref{app:EucgradW}.
\end{proof}

\subsection{Manifold Optimization for ${\bf \Theta}$}

For problem P2.2 with a given ${\bf W}$, we first derive the Euclidean gradient of the cost function $J({\bf \Theta})$. 
Notice that ${\bf \Theta}$ is a structured matrix with all the non-diagonal entries equaling zero.  
For the ease of derivation, we regard the diagonal entries of ${\bf \Theta}$, i.e., ${\bm \theta}\triangleq{\rm diag}({\bf \Theta})$, as the optimization variable. 
The derivative with respect to the structured matrix is introduced in \cite{hjorungnes2011}, accordingly, we have the following lemma.
\begin{lem}\label{lem:gradient}
Assume that $J:{\mathbb C}^{M\times 1}\times {\mathbb C}^{M\times 1} \to {\mathbb C}$ is a scalar function of ${\bm \theta} \in {\mathbb C}^{M\times 1}$ and ${\bm \theta}^* \in {\mathbb C}^{M\times 1}$, and define ${\bf \Theta}\triangleq{\rm diag}({\bm \theta})$. Then, the derivatives of $J$ with respect to ${\bm \theta}^*$ and ${\bf \Theta}^*$ have the following relationship:
\begin{align}\label{eq:gradient}
\frac{\partial J}{\partial {\bm \theta}^*}={\rm diag}\left(\frac{\partial J}{\partial {\bf \Theta}^*}\right).
\end{align}
\end{lem}

\begin{proof}
See Appendix \ref{app:lemgradient}.
\end{proof}

Based on \emph{Lemma \ref{lem:gradient}}, we derive the Euclidean gradient of the cost function of problem P2.2 in the following proposition.
\begin{pro}\label{pro:Eucgradtheta}
The Euclidean gradient of $J({\bm \theta})$ has the following form:
\begin{align}\label{eq:Eucgradtheta}
\frac{\partial J}{\partial{\bm \theta}^*}
&=\frac{2M^2 N_{\rm BS}}{\|{\bf W}\|_F^2}{\rm diag}\left({\bf\widetilde A}^H{\bf \Gamma}^2{\rm diag}\left({\bf\bar y}-{\bf f}\right){\bf\widetilde A}{\bf \Theta}{\bf V}\right),
\end{align}
\noindent where ${\bf V}={\bf A}_{\bf G}{\bf I}\circ\left({\bf\Lambda}{\bf B}_{\bf G}^H{\bf W}
{\bf W}^H{\bf B}_{\bf G}\right){\bf A}_{\bf G}^H$.
\end{pro}

\begin{proof}
See Appendix \ref{app:Eucgradtheta}.
\end{proof}

Next, regarding the non-convex constant modulus constraint in (\ref{eq:P21con}), we apply the Euclidean gradient derived in \emph{Proposition \ref{pro:Eucgradtheta}} over a Riemannnian manifold as follows
\begin{align}\label{eq:RieM}
{\mathcal M} \triangleq
\left\{{\bm \theta}\in{\mathbb C}^{M\times 1}:|{\theta}_m|=1,\ m=0,1,...,M-1\right\}.
\end{align}
\noindent By treating ${\mathbb C}$ as ${\mathbb R}^2$, we define the inner product as $\langle \theta_1,\theta_2\rangle\triangleq \Re[\theta_1\theta_2^*]$ \cite{dujb2019}. Then, the corresponding tangent space ${\mathcal T}_{\bm \theta}{\mathcal M}$ at the point ${\bm \theta}$ can be defined as \cite{Yu2016}
\begin{align}\label{eq:TM}
{\mathcal T}_{\bm \theta}{\mathcal M}\triangleq
\left\{{\bf x}\in{\mathbb C}^{M\times 1}:\Re[{\bf x}\circ{\bm \theta}^*]={\bf 0}\right\}.
\end{align}

Given the definition of the Riemannian manifold, we can employ the CG method to search a local optimum by treating ${\mathcal M}$ as the search space. Before introducing the algorithm, we define two kinds of projections towards the tangent space and the manifold, respectively.

Firstly, we project the Euclidean gradient onto the tangent space ${\mathcal T}_{\bm \theta}{\mathcal M}$, which defines the Riemannian gradient $\nabla_{\bm \theta}J$. Specifically, to project a vector ${\bf d}$ onto the tangent space ${\mathcal T}_{\bm \theta}{\mathcal M}$, we use the following orthogonal projection \cite{dujb2019}:
\begin{align}\label{eq:maptoTM}
{\rm P}_{\bm \theta}\left({\bf d}\right)={\bf d}-\Re\left[{\bf d}\circ{\bm \theta}^*\right]\circ{\bm \theta}.
\end{align}
\noindent We apply this projection to $\partial J({\bm \theta})/\partial {\bm \theta}^*$, and obtain the Riemannian gradient as follows
\begin{align}\label{eq:Riegrad}
\nabla_{\bm \theta}J
={\rm P}_{\bm \theta}\left(\frac{\partial J({\bm \theta})}{\partial {\bm \theta}^*}\right)
=\frac{\partial J({\bm \theta})}{\partial {\bm \theta}^*}
-\Re\left[\frac{\partial J({\bm \theta})}{\partial {\bm \theta}^*}\circ{\bm \theta}^*\right]\circ{\bm \theta}.
\end{align}

After updating the point with the Riemannian gradient, we employ another projection called retraction to project the updated point ${\bf x}$ back onto the manifold, which is defined as
\begin{align}\label{eq:retraction}
{\rm Ret}({\bf x})\triangleq\left[\frac{{x}_0}{|{x}_0|},\frac{{x}_1}{|{x}_1|},...,\frac{{x}_{M-1}}{|{x}_{M-1}|}\right].
\end{align}

With the above definitions of the Riemannian manifold, we propose a manifold based CG method to optimize the phase shifts of the RIS in problem P2.2, which is shown in Algorithm \ref{alg:P21}. 
In the light of \cite[Theorem 4.3.1]{absil2009}, Algorithm \ref{alg:P21} is guaranteed to converge to a critical point.
\begin{algorithm}[h]
  \caption{Riemannian Manifold Based CG Algorithm}
  \label{alg:P21}
  \begin{algorithmic}[1]
    \Require ${\bm \theta}_0$
    \State Set the initial search direction ${\bf d}_0=-\nabla_{{\bm \theta}_0}J$ and $t=0$
    \Repeat
    \State \hspace{-6mm}Choose the step size $\omega_t$ via the Armijo backtracking line search \cite{absil2009}\label{alg:armijo}
    \State \hspace{-6mm}Update the vector: ${\bm \theta}_{t+1}={\rm Ret}({\bm \theta}_{t}+\omega_t{\bf d}_t)$
    \State \hspace{-6mm}Calculate  the Riemannian gradient $\nabla_{{\bm \theta}_{t+1}}J$ according to (\ref{eq:Eucgradtheta}) and (\ref{eq:Riegrad})
    \State \hspace{-6mm}Calculate the Polak-Ribi{\`e}re parameter \cite{absil2009}:\qquad\qquad\qquad
    \Statex ${\mathcal B}_{t+1}=\frac{\left(\nabla_{{\bm \theta}_{t+1}}J\right)^H\left[\nabla_{{\bm \theta}_{t+1}}J-{\rm P}_{{\bm \theta}_{t+1}}\left(\nabla_{{\bm \theta}_{t}}J\right)\right]}{\|\nabla_{{\bm \theta}_{t}}J\|_2^2}$
    \State\hspace{-6mm}Calculate the conjugate direction: \qquad\qquad
    \Statex ${\bf d}_{t+1}=-\nabla_{{\bm \theta}_{t+1}}J+{\mathcal B}_{t+1}{\rm P}_{{\bm \theta}_{t+1}}\left({\bf d}_t\right)$
   \State \hspace{-6mm}$t\leftarrow t+1$
    \Until{convergence.}
  \end{algorithmic}
\end{algorithm}

In Step \ref{alg:armijo} of Algorithm \ref{alg:P21}, given the scalars $l,\gamma\in(0,1)$ and $q>0$, the Armijo step size is $\omega_t=ql^n$, where $n$ is the smallest nonnegative integer that satisfies \cite[Definition 4.2.2]{absil2009}
\begin{align}\label{eq:armijo}
\hspace{-2mm}J({\bm\theta}_t)-J\left({\rm Ret}({\bm\theta}_t+ql^n{\bf d}_t)\right)
\geq -\gamma ql^n  \Re\left[\left(\nabla_{{\bm\theta}_t}J\right)^H {\bf d}_t \right].
\end{align}

\subsection{Alternating Optimization for Pattern Synthesis}
We summarize the alternating optimization for problem P2 in Algorithm \ref{alg:altopt}.
\begin{algorithm}[h]
  \caption{Alternating Optimization for the Manifold Based Pattern Synthesis of the RIS}
  \label{alg:altopt}
  \begin{algorithmic}[1]
    \Require ${\bm \theta}^{(0)}$, ${\bf W}^{(0)}$
    \State Set $t'=0$ 
    \Repeat
    \State \hspace{-6mm}Fix ${\bm \theta}^{(t')}$, optimize ${\bf W}^{(t'+1)}$ using the CG method and \emph{Proposition} \ref{pro:EucgradW}
    \State \hspace{-6mm}Fix ${\bf W}^{(t'+1)}$, optimize ${\bm \theta}^{(t'+1)}$ using Algorithm \ref{alg:P21}
   \State \hspace{-6mm}$t'\leftarrow t'+1$
    \Until{convergence.}
  \end{algorithmic}
\end{algorithm}

In each iteration of Algorithm \ref{alg:altopt}, the value of the cost function $J$ is non-increasing, non-negative, and locally optimal, therefore, the alternating optimization is guaranteed to converge to a locally optimal solution.

\begin{rmk}
The design framework in Algorithm \ref{alg:altopt} considers the general multipath model for the BS-RIS channel, and jointly optimizes the precoder at the BS and the phase shifts of the RIS, which encompasses the special case discussed in \cite{Hai2021} and \cite{He2021} that adopted the LoS BS-RIS channel model to optimize the phase shifts of the RIS.
\end{rmk}

In fact, under the multipath BS-RIS channel, the reflected power pattern cannot be synthesized by directly using the sub-array based method proposed in \cite{Hai2021}.
To reveal the difference, we rewrite the power pattern in (\ref{eq:avepattern}) as follows
\begin{align}\label{eq:avepattern2}
y(\phi)
&=M^2 N_{\rm BS}\sum_{l=0}^{L-1}\chi_l{\bf a}_{\bf H}^H(\phi){\bf \Theta}
{\bf a}_{\bf G}(\phi_l){\bf a}_{\bf G}^H(\phi_l)
{\bf \Theta}^H{\bf a}_{\bf H}(\phi),
\end{align}
\noindent where $\chi_l$ denotes the $l$-th diagonal entry of the matrix ${\bf I}\circ\left({\bf\Lambda}{\bf B}_{\bf G}^H{\bf W}
{\bf W}^H{\bf B}_{\bf G}\right)$, and ${\bf a}_{\bf G}(\phi_l)$ is the $l$-th column of ${\bf A}_{\bf G}$ as defined in (\ref{eq:Gcom}). From (\ref{eq:avepattern2}), we find that the power pattern reflected by the RIS is the superposition of serveral patterns, each of which corresponds to an impinging channel path with the AoA of $\phi_l$. 
Accordingly, we have the following proposition.

\begin{pro}\label{pro:shiftAOA}
	Given the phase shifts at the RIS, changing the AoA of the impinging signal leads to a shift in the AoD of the flat-top beam. Specifically, 
	assume that the RIS reflects the signal from the AoA of $\phi_0$ towards a flat-top region that covers $\mathcal A_1=[\phi_{\rm min},\phi_{\rm max}],\ \phi_{\rm min},\phi_{\rm max}\in[\pi/2,\pi]$.
	Then, changing the AoA of the impinging signal to another angle $\phi_1\ (\phi_1\neq\phi_0)$ results in a shifted flat-top region $\mathcal A'_1=[\phi'_{\rm min},\phi'_{\rm max}]$, where $\phi'_{\rm min}$ and $\phi'_{\rm max}$ are defined by
	\begin{align}\label{eq:A1shift}
	\begin{cases}
	\phi'_{\rm min}=\arccos(\cos\phi_{\rm min}+\xi),\phi'_{\rm max}=\arccos\left(\max\{-1,\cos\phi_{\rm max}+\xi\}\right),&\ \phi_1<\phi_0\\
	\phi'_{\rm min}=\arccos\left(\cos\phi_{\rm min}+\xi\right),\phi'_{\rm max}=\arccos\left(\cos\phi_{\rm max}+\xi\right),&\ \phi_1>\phi_0,
	\end{cases}
	\end{align}
	\noindent where $\xi\triangleq\cos\phi_0-\cos\phi_1$.
\end{pro}

\begin{proof}
	See Appendix \ref{app:shiftAOA}.
\end{proof}

\begin{rmk}\label{rmk:shiftAOA} 
	It is observed in \emph{Proposition \ref{pro:shiftAOA}} that, $\phi'_{\rm min}>\phi_{\rm min}$ and $\phi'_{\rm max}>\phi_{\rm max}$ hold for $\phi_1<\phi_0$, while $\phi'_{\rm min}<\phi_{\rm min}$ and $\phi'_{\rm max}<\phi_{\rm max}$ hold for $\phi_1>\phi_0$, which can be explained by the law of reflection.  Moreover, the range of the shifted flat-top region is not necessarily equal to that of the original flat-top region, i.e., $\phi'_{\rm max}-\phi'_{\rm min}\neq\phi_{\rm max}-\phi_{\rm min}$, while the shape of the flat-top beam holds.
	Note that a similar conclusion was obtained in \cite{Jamali2021} by numerical simulations, while it is strictly proved in this paper. 
\end{rmk}

From (\ref{eq:avepattern2}) and \emph{Proposition \ref{pro:shiftAOA}}, we find that, multiple impinging channel paths with different AoAs turn the reflected pattern towards different AoDs, which leads to interference and scattering. 
Therefore, under the multipath BS-RIS channel, we need to synthesize the superposition of multiple reflected patterns instead of designing a single pattern based on the LoS BS-RIS channel. To this end, we propose the general design framework, i.e., Algorithm \ref{alg:altopt}, in this paper.

Moreover, by using the proposed RIS-assisted broad coverage, we can provide performance enhancements in the mmWave communications of both the broadcast channel and the OFDMA channel.
The detailed discussions of the applications and the corresponding performance analysis are given in the next section.

\section{Applications and Performance Analysis of RIS-assisted Broad Coverage}\label{sec:applicationandperformance}

\subsection{Communications of Broadcast Channels}

As for the communication of broadcast channels in an RIS-assisted massive MIMO-OFDM system, we derive the downlink rate of transmitting multiple streams as follows
\begin{align}\label{eq:rate0}
R=\sum_{k=0}^{N_{\rm c}-1}\log_2\det\left({\bf I}_{N_{\rm UE}}+\frac{p}{\sigma_z^2}{\bf H}_{\rm eq}[k]{\bf W}[k]{\bf W}^H[k]{\bf H}_{\rm eq}^H[k]\right),
\end{align}
\noindent where ${\bf H}_{\rm eq}[k]\triangleq\sqrt{\beta_1 \beta_2}{\bf H}[k]{\bf \Theta}{\bf G}[k]+\sqrt{\beta}{\bf H}_{\rm d}[k]$. It is difficult to derive the analytical ergodic rate from (\ref{eq:rate0}). Therefore, we verify the performance of the proposed broad coverage via numerical simulations in Section \ref{sec:numerical}.
Nevertheless, we derive an analytical expression of the average received signal power at the user by ignoring the blocked BS-UE direct channel in the following theorem.
\begin{thm}\label{thm:flattop}
Assume that the perfect power pattern is synthesized in accordance with (\ref{eq:targetf}), i.e., $y(\phi)=f(\phi)$, the side lobe $f_S=0$, and the rectangular roll-off $\mathcal{A}_3=\emptyset$, each user within the broad coverage receives an average power, across all the antennas and subcarriers, as follows
\begin{align}\label{eq:UEpower}
{\mathbb E}\{|r|^2\}=p \beta_1 \beta_2 f_{\rm M}\sum_{\phi_q\in\mathcal{A}_1}\mathbb{E}\left\{|\alpha_q|^2\right\}+\sigma_z^2,
\end{align}
\noindent where $\mathcal{A}_1$ is the flat-top region of the power pattern defined in Section \ref{subsec:formulation}, and $\sum_{\phi_q\in\mathcal{A}_1}\mathbb{E}\left\{|\alpha_q|^2\right\}$ denotes the average power of all the channel paths within the broad coverage $\mathcal{A}_1$. 
Furthermore, without loss of generality, we assume that the channel power of NLoS paths is uniformly distributed over $[0,\pi]$, therefore, we find that each user within the broad coverage receives a constant average power as follows
\begin{align}\label{eq:UEpower2}
{\mathbb E}\{|r|^2\}=p \beta_1 \beta_2 f_{\rm M}\left[\frac{K}{K+1}+\frac{|\mathcal{A}_1|}{\pi(K+1)}\right]+\sigma_z^2,
\end{align}
\noindent where $K$ is the K-factor of the RIS-UE channel, and $|\mathcal{A}_1|$ denotes the beamwidth of the power pattern reflected by the RIS.
\end{thm}

\begin{proof}
See Appendix \ref{app:lemflattop}.
\end{proof}

\begin{rmk}\label{rmk:UEpower}
The analytical expression in (\ref{eq:UEpower}) indicates that, given an idealized power pattern, the received signal power of the user depends on the power gain reflected by the RIS, i.e., $f_{\rm M}$. Then, we give the power scaling law in the following proposition.
\end{rmk}

\begin{pro}\label{pro:fm}
The received power of the user and the power of the flat-top beam both increase linearly with the number of RIS reflecting elements $M$, and decrease linearly with the beamwidth of the power pattern reflected by the RIS.
\end{pro}

\begin{proof}
See Appendix \ref{app:fm}.
\end{proof}

\subsection{Communications of OFDMA Channels}

We consider applying the quasi-static broad coverage to the scenario of OFDMA channels, where each user is assigned to a respective set of subcarriers. For multiuser multi-carrier systems, it is intractable to design a common set of RIS phase shifts for serving multiple users at each respective subcarrier. Moreover, the overhead of channel estimation and feedback is large when the phase shifts of the RIS are optimized under the instantaneous CSI.
Fortunately, the proposed broad coverage can provide quasi-static services for multiple users simultaneously, which only needs a reduced overhead of the statistical CSI.
Specifically, by using the LoS BS-RIS channel, we can decouple the optimizations of the precoder at the BS and the phase shifts of the RIS \cite{He2021}, which facilitates an independent design for each user at the respective subcarrier.
Accordingly, we consider an RIS-assisted massive MIMO-OFDMA system with multiple single-antenna users.
The downlink rate of the considered system is given by
\begin{align}\label{eq:rateMISO}
R'=&\sum_{u=0}^{U-1}\sum_{k\in\{{\mathcal K}_u\}}\log_2\left[1+\frac{p}{\sigma_z^2}\left|\left(\sqrt{\beta_1 \beta_2}{\bf h}_u^H[k]{\bf \Theta}{\bf G}_0[k]+\sqrt{\beta}{\bf h}_{{\rm d}u}^H[k]\right){\bf w}[k]\right|^2\right],
\end{align}
\noindent where $u$ denotes the index of user, $U$ is the number of users, ${\mathcal K}_u$ denotes the set of subcarriers allocated to the $u$-th user, ${\bf h}_u^H[k]$ is the channel between the RIS and the $u$-th user, ${\bf G}_0[k]$ is the LoS BS-RIS channel, and ${\bf h}_{{\rm d}u}^H[k]$ is the direct channel between the BS and the $u$-th user. Accordingly, we have
\begin{align}
{\bf h}_u^H[k]&\triangleq\sqrt{M}\sum_{q_u=0}^{Q-1}\delta_{q_u}[k]{\bf a}_{\bf H}^H(\phi_{q_u}),\label{eq:hr}\\
{\bf G}_0[k]&\triangleq\sqrt{N_{\rm BS}M}\eta{\bf a}_{\bf G}(\phi_{0}){\bf b}_{\bf G}^H(\psi_{0}),\label{eq:GL}\\
{\bf h}_{{\rm d}u}^H[k]&\triangleq\sqrt{N_{\rm BS}}\sum_{q'_u=0}^{Q'-1}\delta_{q'_u}[k]{\bf b}_{\bf G}^H(\psi_{q'_u}),\label{eq:hd}
\end{align}
\noindent where $\eta$ is the complex phase of the LoS channel, 
$\delta_{q_u}[k]$ and $\delta_{q'_u}[k]$ are the frequency-domain complex gains of the $q_u$-th path of ${\bf h}_u^H[k]$ and the $q'_u$-th path of ${\bf h}_{{\rm d}u}^H[k]$, respectively, 
and $\phi_{q_u}$ and $\psi_{q'_u}$ are the AoDs of the $q_u$-th path of ${\bf h}_u^H[k]$ and the $q'_u$-th path of ${\bf h}_{{\rm d}u}^H[k]$, respectively.
In detail, we define the frequency-domain complex gains, i.e., $\delta_{q_u}[k]$ and $\delta_{q'_u}[k]$, as follows 
\begin{align}
\delta_{q_u}[k]&\triangleq\sum_{n=0}^{N_{\rm c}-1}\alpha_{q_u}g(nT-\tau_{q_u})e^{-\jmath \frac{2\pi kn}{N_{\rm c}}}=\alpha_{q_u}e^{-\jmath \frac{2\pi kn_{q_u}}{N_{\rm c}}},\label{eq:hrdelta}\\
\delta_{q'_u}[k]&\triangleq\sum_{n=0}^{N_{\rm c}-1}\alpha_{q'_u}g(nT-\tau_{q'_u})e^{-\jmath \frac{2\pi kn}{N_{\rm c}}}=\alpha_{q'_u}e^{-\jmath \frac{2\pi kn_{q'_u}}{N_{\rm c}}},\label{eq:hddelta}
\end{align}
\noindent where $\alpha_{q_u}$ and $\alpha_{q'_u}$ are the time-domain complex gains of the ${q_u}$-th path of the RIS-UE channel and the ${q'_u}$-th path of the BS-UE direct channel, respectively, $\tau_{q_u}$ and $\tau_{q'_u}$ are the transmission delays of the ${q_u}$-th path of the RIS-UE channel and the ${q'_u}$-th path of the BS-UE direct channel, respectively, 
and the indices $n_{q_u}$ and $n_{q'_u}$ satisfy $n_{q_u}T<\tau_{q_u}\leq(n_{q_u}+1)T$ and $n_{q'_u}T<\tau_{q'_u}\leq(n_{q'_u}+1)T$, respectively. 
We assume that the time-domain channel gains, i.e., $\alpha_{q_u}$ and $\alpha_{q'_u}$, are all uncorrelated and Gaussian distributed with normalized channel powers, i.e., $\sum_{{q_u}=0}^{Q-1}{\mathbb E}\{|\alpha_{q_u}|^2\}=1$ and $\sum_{{q'_u}=0}^{Q'-1}{\mathbb E}\{|\alpha_{q'_u}|^2\}=1$.

For any given ${\bf \Theta}$ in (\ref{eq:rateMISO}), the optimal precoder that maximizes $R'$ corresponds to the maximum ratio transmission (MRT) \cite{wu2019intelligent}, i.e.,
\begin{align}\label{eq:MRT}
{\bf w}[k]=\frac{\sqrt{\beta_1\beta_2}{\bf G}_0^H[k]{\bf \Theta}^H{\bf h}_u[k]+\sqrt{\beta}{\bf h}_{{\rm d}u}[k]}{\left\|\sqrt{\beta_1\beta_2}{\bf G}_0^H[k]{\bf \Theta}^H{\bf h}_u[k]+\sqrt{\beta}{\bf h}_{{\rm d}u}[k]\right\|_2}.
\end{align}
\noindent Then, we derive the rate in (\ref{eq:rateMISO}) as follows
\begin{align}\label{eq:rateMISO2}
R'=&\sum_{u=0}^{U-1}\sum_{k\in\{{\mathcal K}_u\}}\log_2\left(1+\frac{p}{\sigma_z^2}\left\|\sqrt{\beta_1 \beta_2}{\bf h}_u^H[k]{\bf \Theta}{\bf G}_0[k]+\sqrt{\beta}{\bf h}_{{\rm d}u}^H[k]\right\|_2^2\right).
\end{align}
From (\ref{eq:rateMISO2}), we obtain the analytical rate of the considered system in the following theorem.
\begin{thm}\label{thm:rateOFDMA}
In an RIS-assisted multiuser massive MIMO-OFDMA system, if the perfect power pattern of the RIS is synthesized in accordance with (\ref{eq:targetf}) and based on the LoS BS-RIS channel \cite{He2021}, the downlink rate can be approximated in a closed form as follows
\begin{align}\label{eq:rateOFDMA}
{\mathbb E}\{R'\}\approx
N_{\rm c}\log_2\left\{1+\frac{p \beta_1 \beta_2  f_{\rm M}}{\sigma_z^2}\left[\frac{K}{K+1}+\frac{|\mathcal{A}_1|}{\pi(K+1)}\right]+\frac{p\beta N_{\rm BS}}{\sigma_z^2}\right\},
\end{align}
where $f_{\rm M}$ is the power of the flat-top beam defined in (\ref{eq:targetf}), $K$ is the K-factor of the RIS-UE channel, and $|\mathcal{A}_1|$ denotes the beamwidth of the power pattern.
\end{thm}

\begin{proof}
See Appendix \ref{app:thmOFDMA}.
\end{proof}

\begin{rmk}\label{rmk:fm}
When the power gain $f_{\rm M}$ is large, the analytical rate derived in (\ref{eq:rateOFDMA}) increases logarithmically with $f_{\rm M}$. 
Specifically, we find from \emph{Proposition \ref{pro:fm}} that, when the number of RIS reflecting elements $M\to\infty$, the analytical rate of the considered OFDMA system increases logarithmically with $M$ and decreases logarithmically with the beamwidth of the flat-top pattern reflected by the RIS.
\end{rmk}

\begin{rmk}\label{rmk:twokindsofw}
The proposed RIS-assisted quasi-static broad coverage is well compatible with the current massive MIMO-OFDM system.
Firstly, since the designs of the precoder at the BS and the phase shifts of the RIS are decoupled, the communications of the OFDMA channel and the broadcast channel can be multiplexed in the frequency domain by applying different precoders to different subcarriers while employing the same set of phase shifts at the RIS.
Furthermore, we can perform the long-term update of the RIS phases with an off-line codebook such that we can improve the system performance with an unchanged frame structure, where no more time slot of training or feedback is needed for the instantaneous RIS design.
\end{rmk}

\section{Numerical Results}\label{sec:numerical}

In this section, we present the numerical simulation of the RIS-assisted broad coverage in the communications of the broadcast channel and the OFDMA channel, respectively. As shown in Fig. \ref{fig:overview}, we assume that a BS with $N_{\rm BS}=64$ antennas is at the origin point, an RIS is at $(190,10)\ \text{m}$, and multiple users are located around the point $(200,0)$ m within a broad coverage over $[\phi_{\rm min},\phi_{\rm max}]$, $\phi_{\rm min},\phi_{\rm max}\in(0,\pi)$. 
For a beam covering $[\phi_{\rm min},\phi_{\rm max}]$, the beamwidth is $2\phi_{0.5}=\phi_{\rm max}-\phi_{\rm min}$, and the steering angle is $\phi_{\rm c}=(\phi_{\rm max}+\phi_{\rm min})/2$. 
The power pattern is oversampled by a factor of $\kappa=10$.
We define the path loss as ${\rm PL}\triangleq{\rm PL}_0+10\zeta\log_{10}(d)$, where ${\rm PL}_0=30$ dB is the path loss at the reference distance of $1$ m, $\zeta$ is the path loss exponent (PLE), and $d$ is the distance. Then, the large-scale fading factor $\beta=10^{-0.1{\rm PL}}$. The PLEs of the BS-RIS channel, the RIS-UE channel, and the BS-UE direct channel are 2, 2.2, and 3.5, respectively.
The transmit power of the BS is $p=20$ dBm, and the noise power of the user is ${\sigma_z^2}=-80$ dBm.
The number of subcarriers is $N_{\rm c}=64$ and the length of the OFDM cyclic prefix (CP) is $L_{\rm CP}=8$. Then, the overhead of CP is considered for the rate $R$ by calculating $\frac{N_{\rm c}}{N_{\rm c}+L_{\rm CP}}R$.

\subsection{Broadcast Channel}

\begin{figure}[!t]
\centering
\includegraphics[width=13pc]{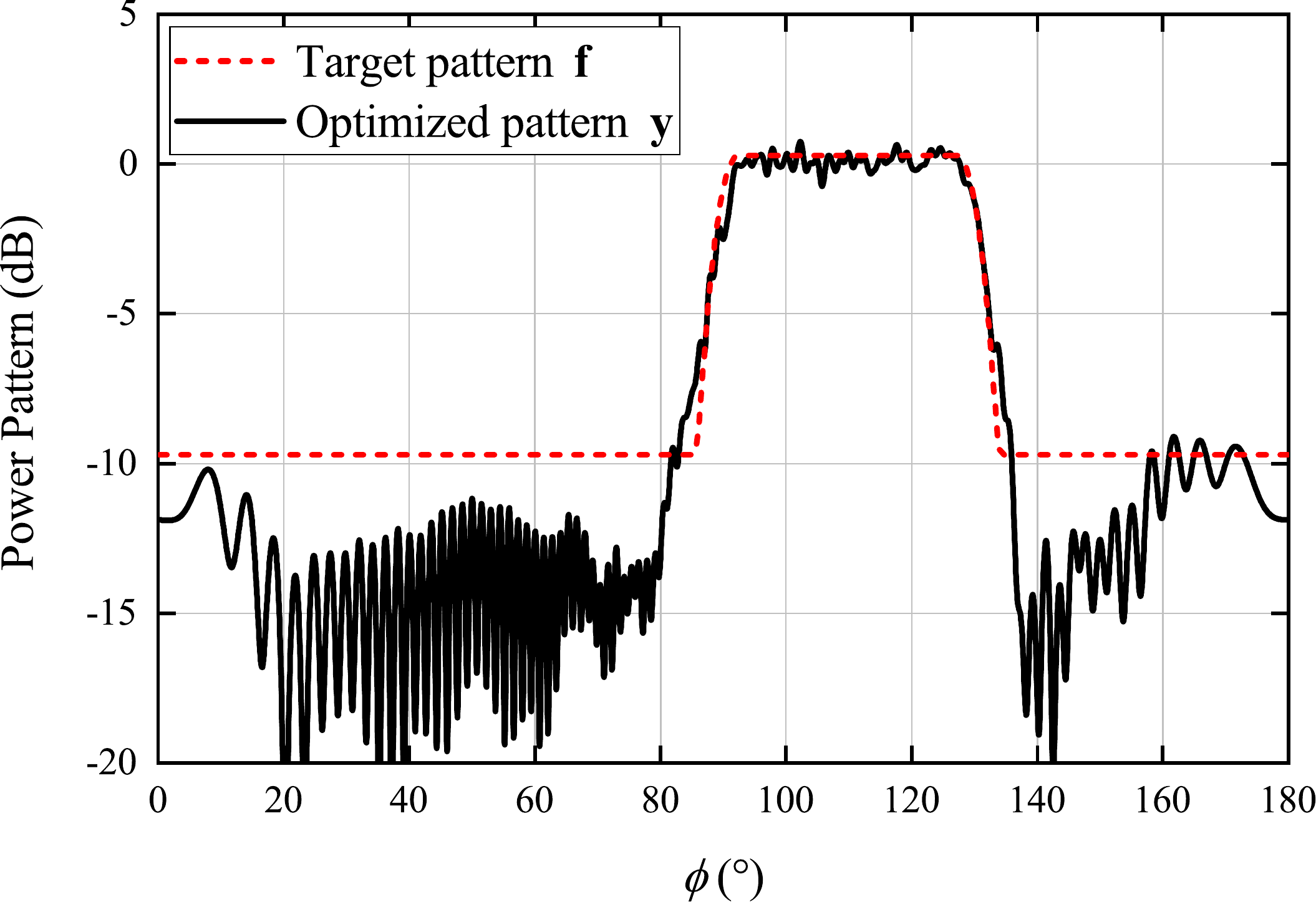}
\caption{Power pattern reflected by the RIS.}
\label{fig:pattern}
\end{figure}

Firstly, we evaluate the performance of Algorithm \ref{alg:altopt}.
In Fig. \ref{fig:pattern}, we synthesize the reflected power pattern of an RIS with $M=100$ elements for transmitting $N_{\rm d}=4$ streams under the BS-RIS mmWave channel with $L=5$ channel paths.
In the considered multipath channel, we assume that the LoS BS-RIS channel path is as strong as the NLoS channel paths, whose profile is assumed to be uniform.
It is shown in Fig. \ref{fig:pattern} that the optimized power pattern provides a broad coverage from  $90^{\circ}$ to $140^{\circ}$ with a power fluctuation of 1.5 dB.

\begin{figure}[!t]
	\centering
	\includegraphics[width=13pc]{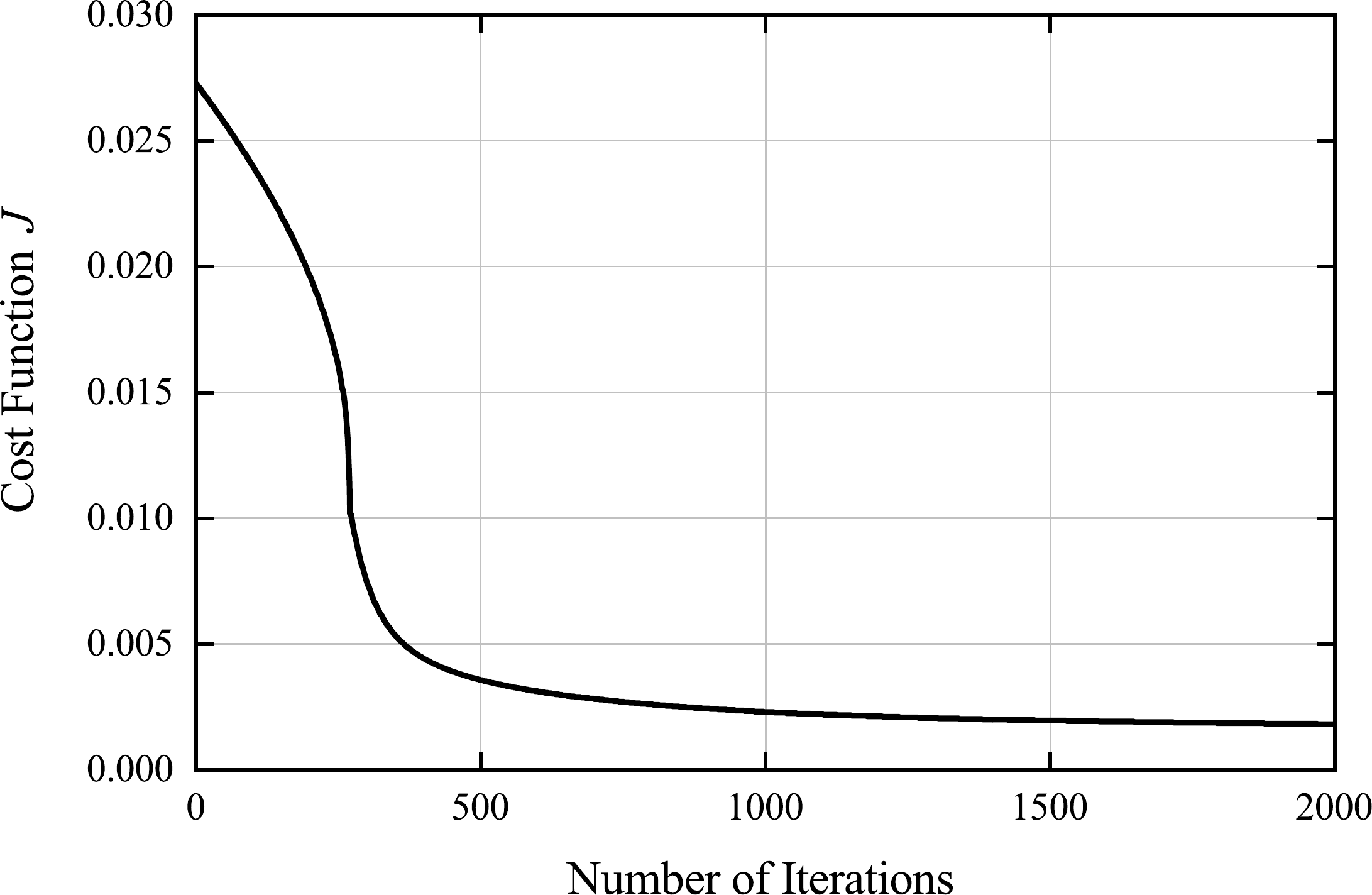}
	\caption{Convergence behavior of the cost function $J$.}
	\label{fig:convergence}
\end{figure}

In Fig. \ref{fig:convergence}, we present the convergence behavior of Algorithm \ref{alg:altopt}, where the cost function $J$ is plotted versus the number of iterations. Fig. \ref{fig:convergence} depicts that the proposed algorithm converges as proved in \cite[Theorem 4.3.1]{absil2009}.

\begin{figure}[!t]
\centering
\includegraphics[width=13pc]{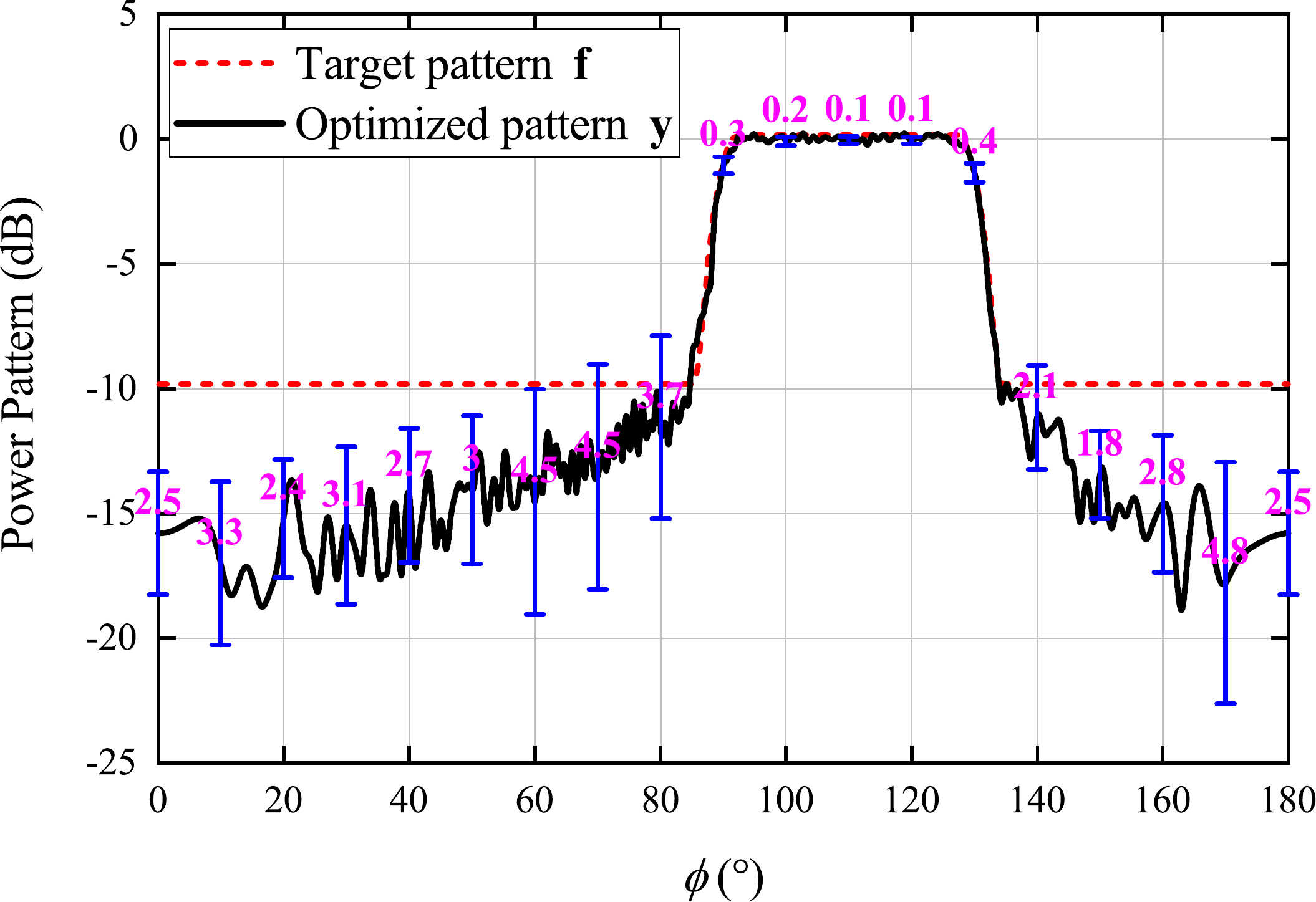}
\caption{Average results of the power patterns under 50 different multipath BS-RIS channels.}
\label{fig:avepattern}
\end{figure}

To present the validity of the proposed algorithm, we further synthesize the reflected power pattern of the RIS over 50 randomly generated multipath BS-RIS channels. In Fig. \ref{fig:avepattern}, we average these 50 power patterns and mark the standard deviation of the power gain at different angles.
It is found that the standard deviation is small within the flat-top region, which indicates that the proposed algorithm is effective over different mmWave multipath channels.

Moreover, we evaluate the spectral efficiency of the considered RIS-assisted system in Fig. \ref{fig:CDFBC} by simulating the cumulative distribution function (CDF) of the downlink rate derived in (\ref{eq:rate0}). Within the broad coverage designed in Fig. \ref{fig:pattern}, we randomly generate 1280 users to display the CDF of the downlink rate, where each user is assigned to one subcarrier of the OFDM symbol. For the RIS-UE channel, we assume the uniform profile with four mmWave paths for the NLoS channel, and set the K-factor to $10$ dB. 
For the BS-UE direct channel, we assume that the LoS channel path is blocked, and the NLoS mmWave channel paths follow a uniform profile.
The AoAs and AoDs of the mmWave paths are assumed to be uniformly distributed.
Moreover, all the mmWave channels are randomly generated for 500 realizations.
In Fig. \ref{fig:CDFBC}, we also display the downlink rate of the random RIS phases and the downlink rate with no RIS, where the BS directly provides a broad coverage for the users.
It is observed that the broad coverage designed by the proposed algorithm obtains the highest downlink rate, which verifies its effectiveness.

\begin{figure}[!t]
	\centering
	\includegraphics[width=13pc]{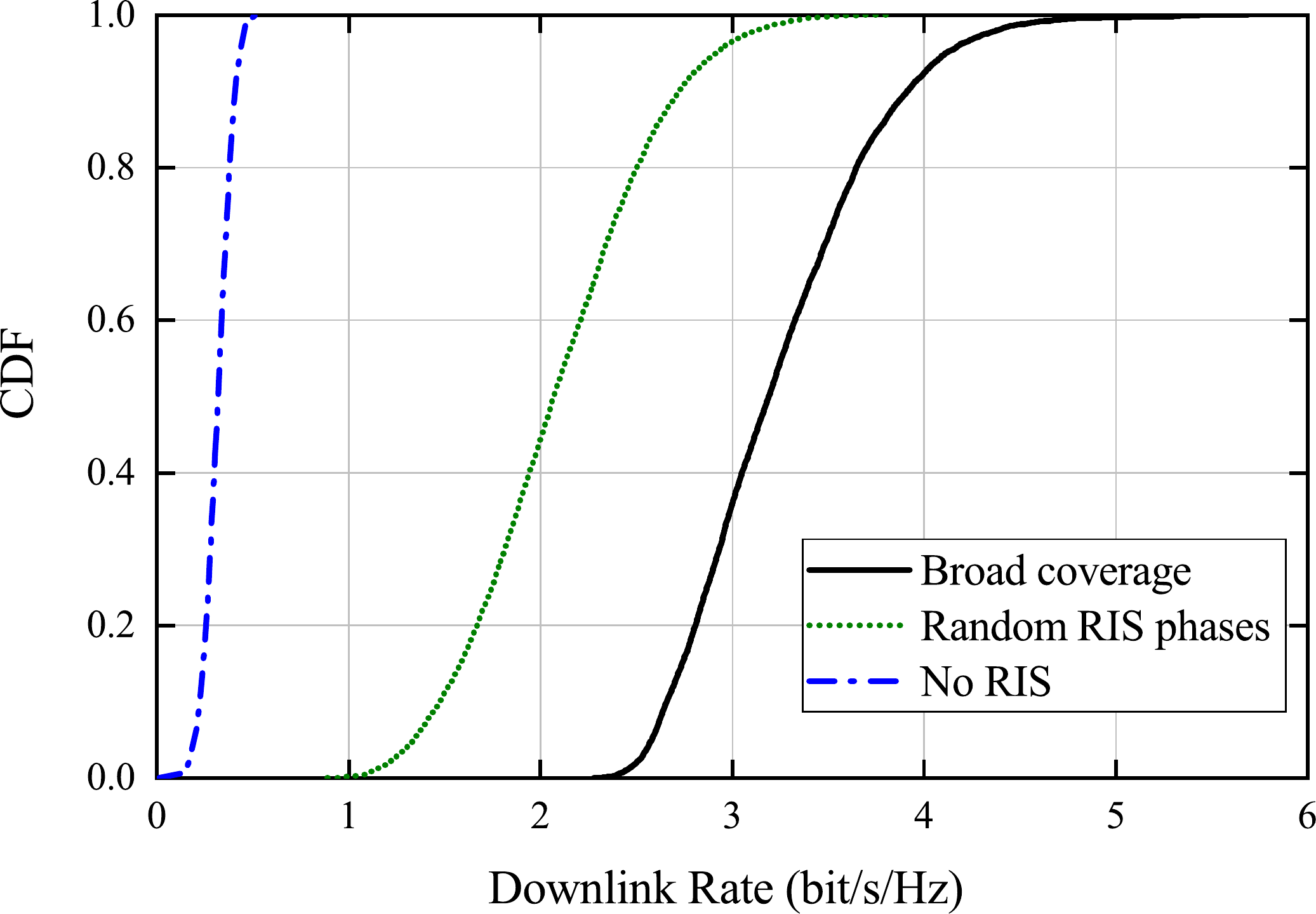}
	\caption{CDF of the downlink rate in the communication of broadcast channel.}
	\label{fig:CDFBC}
\end{figure}

\subsection{OFDMA Channel}

\begin{figure}[!t]
\centering
\includegraphics[width=13pc]{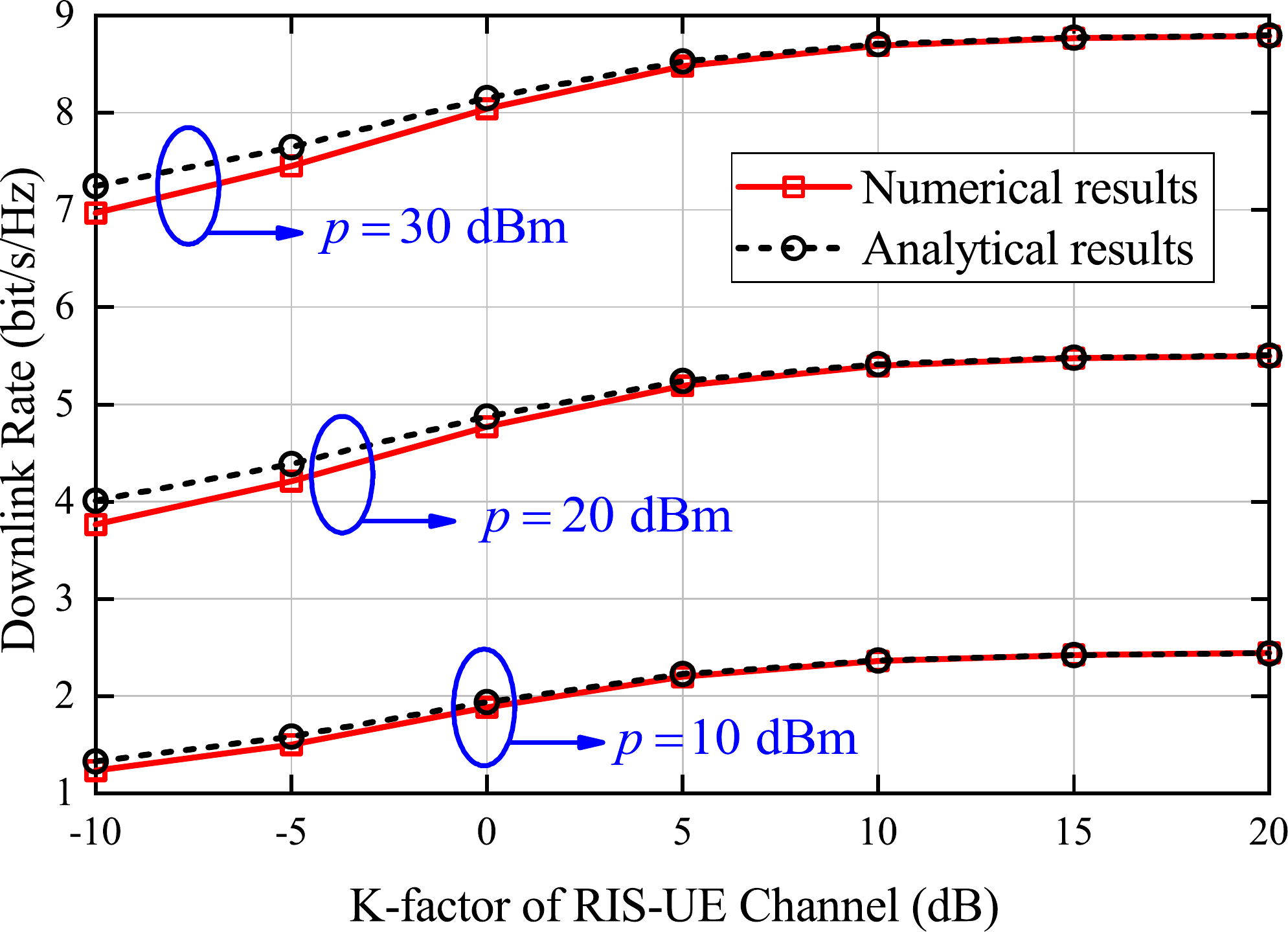}
\caption{Downlink rate under different K-factors of the RIS-UE channel.}
\label{fig:ratevsK}
\end{figure}

In this subsection, we present the performance of the RIS-assisted broad coverage in OFDMA channels.
In the considered OFDMA system, we design a broad coverage from $90^{\circ}$ to $120^{\circ}$ by assuming an RIS with $M=200$ reflecting elements.

Firstly, we present the ergodic downlink rate under different channels and different transmit powers.
In Fig. \ref{fig:ratevsK}, we compare the analytical expression derived in (\ref{eq:rateOFDMA}) with the numerical results based on the instantaneous rate in (\ref{eq:rateMISO2}).
We assume a dominant LoS channel between the BS and the RIS.
For the RIS-UE channel, we assume the uniform profile with three NLoS mmWave channel paths, and change the K-factor from $-10$ dB to $20$ dB.
For the BS-UE direct channel, we assume that the LoS channel path is blocked, and the NLoS mmWave channel paths follow a uniform profile.
The AoAs and the AoDs of mmWave channel paths are assumed to be uniformly distributed.
We average the downlink rate of the considered OFDMA system over 1000 randomly generated channels, where the transmit power of the BS changes from $10$ dBm to $30$ dBm.
It is found in Fig. \ref{fig:ratevsK} that the analytical downlink rate derived in (\ref{eq:rateOFDMA}) matches the numerical results well.
At the small K-factors, the approximation errors between the analytical results and the numerical results are slightly larger due to the larger differences among the channel powers at the BS antennas, which leads to a lager approximation error in (\ref{eq:rateMISOderive}) \cite{QIZHANG2014}.

\begin{figure}[!t]
\centering
\includegraphics[width=13pc]{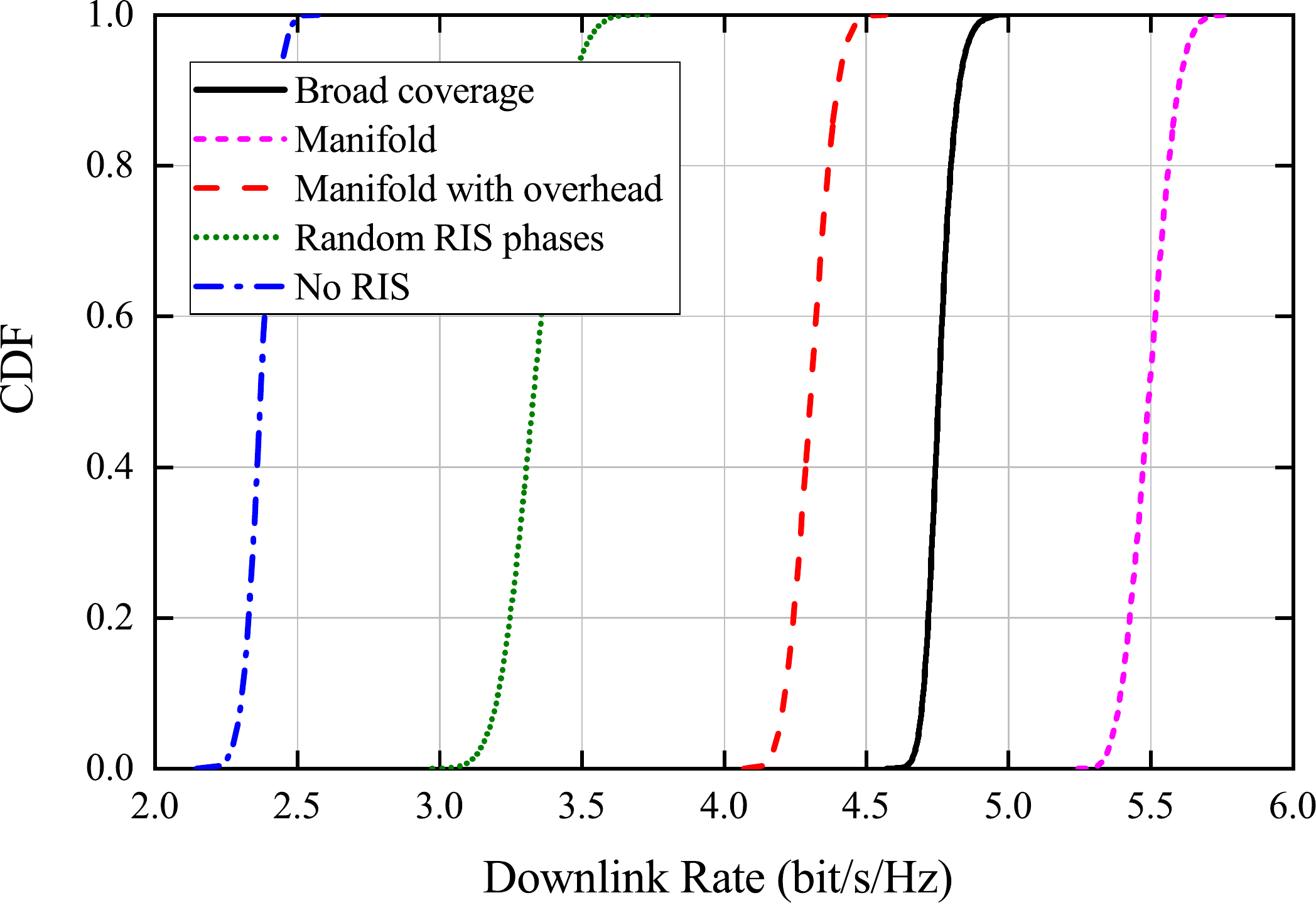}
\caption{CDF of the downlink rate for the OFDMA channel.}
\label{fig:CDFOFDMA}
\end{figure}

In Fig. \ref{fig:CDFOFDMA}, we display the CDF of the downlink rate of the considered OFDMA system. For the BS-RIS channel and the RIS-UE channel, we assume uniform profiles with four mmWave paths for the NLoS channels and set the K-factors to $10$ dB.
The AoAs and AoDs of the mmWave channel paths are assumed to be uniformly distributed.
We compare the performance of the proposed algorithm with the manifold method \cite{Yu2019}, which maximizes the downlink rate by using the instantaneous CSI.
From Fig. \ref{fig:CDFOFDMA}, we find that, compared with the scenarios of the random RIS phases and no RIS, both the proposed broad coverage and the manifold method improve the downlink system rate.
It is apparent that the manifold method with the perfect instantaneous CSI outperforms the broad coverage with the statistical CSI. However, the manifold method loses its advantage if we consider the estimation error of the instantaneous CSI and the overhead of channel estimation.
In Fig. \ref{fig:CDFOFDMA}, we assume an estimation error with the normalized mean squared error (NMSE) of 0.2 for the manifold method, and assume that $20\%$ of the coherence time is used for the channel estimation in a relative short coherence block length.
For instance, in a coherence block length of ten OFDM symbols, two symbols are allocated for the channel estimation, whose overhead is calculated by adopting the grouping method introduced in \cite{OFDM9039554} with a grouping ratio of $1/200$.
Considering the estimation error and the overhead of channel estimation, it is found in Fig. \ref{fig:CDFOFDMA} that the downlink rate of the manifold method is lower than that of the proposed quasi-static  broad coverage.

\section{Conclusion}\label{sec:conclusion}

In this paper, we design the quasi-static broad coverage at the RIS to provide reliable mmWave communications with reduced overhead instead of optimizing the phase shifts of the RIS based on the instantaneous CSI.
We propose a general design framework to synthesize the flat-top pattern with small power fluctuations.
For the communication of broadcast channels, we generalize the broad coverage of the single transmit stream to the scenario of multiple streams.
For the communication of OFDMA channels, we find that, by taking into account the overhead of channel estimation, the proposed quasi-static broad coverage can achieve a higher rate than the optimal RIS design that adopts the instantaneous CSI.
Moreover, we derive the analytical downlink rate of the considered OFDMA system, which scales logarithmically with the power gain reflected by the RIS.

\appendices
\section{Proof of \emph{Proposition \ref{pro:EucgradW}}}\label{app:EucgradW}

By defining $c\triangleq M^2 N_{\rm BS}/\|{\bf W}\|_F^2$,
${\bf V}\triangleq{\bf A}_{\bf G}{\bf I}\circ\left({\bf\Lambda}{\bf B}_{\bf G}^H{\bf W}
{\bf W}^H{\bf B}_{\bf G}\right){\bf A}_{\bf G}^H$, 
${\bf Y}\triangleq{\bf \Gamma}{\bf\widetilde A}{\bf \Theta}{\bf V}{\bf \Theta}^H{\bf\widetilde A}^H$, 
and ${\bf F}\triangleq{\rm diag}\left({\bf \Gamma}{\bf f}\right)$, we derive the cost function of problem P2 as follows
\begin{align}\label{Jderive}
J&=\left\|{\bf \Gamma}\left({\bf f}-{\bf \bar y}\right)\right\|_2^2
=\left\|{\bf \Gamma}{\bf f}\right\|_2^2
+c^2{\rm tr}\left({\bf Y}\circ{\bf Y}\right)
-2c{\rm tr}\left({\bf F}{\bf Y}^H\right).
\end{align}

Firstly, we obtain the derivatives of $1/\|{\bf W}\|_F^4$ and $1/\|{\bf W}\|_F^2$ as follows
\begin{align}\label{eq:term1}
\frac{d\left(\frac{1}{\|{\bf W}\|_F^4}\right)}{d{\bf W}^*}
=\frac{-2{\bf W}}{\left[{\rm tr}\left({\bf W}{\bf W}^H\right)\right]^3},\ \ \ \ \ 
\frac{d\left(\frac{1}{\|{\bf W}\|_F^2}\right)}{d{\bf W}^*}
=\frac{-{\bf W}}{\left[{\rm tr}\left({\bf W}{\bf W}^H\right)\right]^2}.
\end{align}

Then, we have the derivative of ${\rm tr}\left({\bf Y}\circ{\bf Y}\right)$ as follows
\begin{align}
\frac{\partial\left({\rm tr}\left({\bf Y}\circ{\bf Y}\right)\right)}{\partial{\bf W}^*}
&\mathop{=}^{(a)}\frac{2{\rm tr}\left(\left({\bf I}\circ{\bf Y}^T\right){\bf \Gamma}{\bf\widetilde A}{\bf \Theta}{\bf A}_{\bf G}{\bf I}\circ\left({\bf\Lambda}{\bf B}_{\bf G}^H{\bf W}
	d\left({\bf W}^H\right){\bf B}_{\bf G}\right){\bf A}_{\bf G}^H{\bf \Theta}^H{\bf\widetilde A}^H\right)}{d {\bf W}^*}\nonumber\\
&\mathop{=}^{(b)}\frac{2{\rm tr}\left({\bf I}\circ\left[{\bf A}_{\bf G}^H{\bf \Theta}^H{\bf\widetilde A}^H\left({\bf I}\circ{\bf Y}^T\right){\bf \Gamma}{\bf\widetilde A}{\bf \Theta}{\bf A}_{\bf G}\right]{\bf\Lambda}{\bf B}_{\bf G}^H{\bf W}
	d\left({\bf W}^H\right){\bf B}_{\bf G}\right)}{d {\bf W}^*}\nonumber\\
&=2{\bf B}_{\bf G}{\bf I}\circ\left[{\bf A}_{\bf G}^H{\bf \Theta}^H{\bf\widetilde A}^H\left({\bf I}\circ{\bf Y}^T\right){\bf \Gamma}{\bf\widetilde A}{\bf \Theta}{\bf A}_{\bf G}\right]{\bf\Lambda}{\bf B}_{\bf G}^H{\bf W},\label{rq:term3}
\end{align}
\noindent where the equations $(a)$ and $(b)$ hold from the following identity \cite{horn1994topics}:
\begin{align}\label{eq:identity}
{\rm tr}\left({\bf A}({\bf B}\circ{\bf C})\right)
={\rm tr}\left(({\bf A}\circ{\bf B}^T){\bf C}\right).
\end{align}

Moreover, the derivative of ${\rm tr}\left({\bf F}{\bf Y}^H\right)$ is derived as follows
\begin{align}
\frac{\partial\left({\rm tr}\left({\bf F}{\bf Y}^H\right)\right)}{\partial{\bf W}^*}
&=\frac{{\rm tr}\left({\bf F}{\bf\widetilde A}{\bf \Theta}{\bf A}_{\bf G}{\bf I}\circ\left[{\bf\Lambda}{\bf B}_{\bf G}^H{\bf W}
	d\left({\bf W}^H\right){\bf B}_{\bf G}\right]{\bf A}_{\bf G}^H{\bf \Theta}^H{\bf\widetilde A}^H{\bf \Gamma}\right)}{d{\bf W}^*}\nonumber\\
&={\bf B}_{\bf G}{\bf I}\circ\left({\bf A}_{\bf G}^H{\bf \Theta}^H{\bf\widetilde A}^H{\bf \Gamma}{\bf F}{\bf\widetilde A}{\bf \Theta}{\bf A}_{\bf G}\right){\bf\Lambda}{\bf B}_{\bf G}^H{\bf W}.\label{eq:term4}
\end{align}

Then, from (\ref{Jderive}), we derive the gradient of the cost function $J$ as follows
\begin{align}\label{eq:1}
\frac{\partial J}{\partial{\bf W}^*}
=&M^4 N_{\rm BS}^2\frac{d\left(\frac{1}{\|{\bf W}\|_F^4}\right)}{d{\bf W}^*}
{\rm tr}\left({\bf Y}\circ{\bf Y}\right)
+\left(\frac{M^2 N_{\rm BS}}{\|{\bf W}\|_F^2}\right)^2\frac{d\left({\rm tr}\left({\bf Y}\circ{\bf Y}\right)\right)}{d{\bf W}^*}\nonumber\\
&-2M^2 N_{\rm BS}\frac{d\left(\frac{1}{\|{\bf W}\|_F^2}\right)}{d{\bf W}^*}
{\rm tr}\left({\bf F}{\bf Y}^H\right)
-\frac{2M^2 N_{\rm BS}}{\|{\bf W}\|_F^2}
\frac{d\left({\rm tr}\left({\bf F}{\bf Y}^H\right)\right)}{d{\bf W}^*}.
\end{align}
Substituting (\ref{eq:term1})-(\ref{eq:term4}) into (\ref{eq:1}), it yields (\ref{eq:gradW}).

\section{Proof of \emph{Lemma \ref{lem:gradient}}}\label{app:lemgradient}

Since there exist dependencies among the entries of the diagonal matrix ${\bf \Theta}^*$ \cite{hjorungnes2011}, we define a parameterization function ${\bf F}({\bm \theta}^*)\triangleq {\bf \Theta}^*$, and  regard the cost function $J$ as the following composed function:
\begin{align}\label{eq:Jv}
J({\bm \theta}^*)=J({\bf \widetilde \Theta}^*){\big|}_{{\bf \widetilde \Theta}^*={\bf \Theta}^*={\bf F}({\bm \theta}^*)}
=J({\bf F}({\bm \theta}^*)),
\end{align}
\noindent where ${\bf \widetilde \Theta}^* \in {\mathbb C}^{M \times  M}$ is a matrix with independent entries \cite{hjorungnes2011}.

To unify the format of derivative, we follow the definition in \cite[Definition 3.2]{hjorungnes2011}, and give the derivative of ${\bf F}$ with respect to the complex matrix ${\bf Z}^*$ as follows
\begin{align}\label{definederivative}
{\mathcal D}_{{\bf Z}^*}{\bf F}({\bf Z},{\bf Z}^*)=\frac{\partial {\rm vec}({\bf F}({\bf Z},{\bf Z}^*))}{\partial {\rm vec}^T({\bf Z}^*)}.
\end{align}

Then, we use the chain rule of complex-valued derivatives \cite{hjorungnes2011}, and obtain
\begin{align}\label{eq:chainrule}
{\mathcal D}_{{\bm \theta}^*}J({\bm \theta}^*)=\left({\mathcal D}_{{\bf\widetilde \Theta}^*}J({\bf \widetilde \Theta}^*){\big|}_{{\bf \widetilde \Theta}^*={\bf \Theta}^*={\bf F}({\bm \theta}^*)}\right)
{\mathcal D}_{{\bm \theta}^*}{\bf F}.
\end{align}
As for ${\mathcal D}_{{\bm \theta}^*}{\bf F}$, we find that
\begin{align}\label{eq:Fv}
d{\rm vec}({\bf \Theta}^*)=d{\rm vec}({\bf F}({\bm \theta}^*))={\bf L}_{\rm D}d{\bm \theta}^*,
\end{align}
\noindent where ${\bf L}_{\rm D}$ is an $M^2\times M$ matrix that places the diagonal entries of ${\bf A} \in {\mathbb C}^{M \times  M}$ onto ${\rm vec}({\bf A})$ \cite{hjorungnes2011}, i.e., ${\rm vec}\left({\bf I}\circ{\bf A}\right)={\bf L}_{\rm D}{\rm diag}({\bf A})$.
Therefore, we have ${\mathcal D}_{{\bm \theta}^*}{\bf F}={\bf L}_{\rm D}$.

We rewrite both sides of (\ref{eq:chainrule}) by using the formal derivative given in (\ref{definederivative}), and obtain
\begin{align}\label{eq:chainrule2}
\left[{\rm vec}\left(\frac{\partial J}{\partial {\bm \theta}^*}\right) \right]^T
&=\left[{\rm vec}\left(\frac{\partial J}{\partial {\bf\widetilde \Theta}^*}\right)
{\bigg|}_{{\bf \widetilde \Theta}^*={\bf \Theta}^*} \right]^T {\bf L}_{\rm D}
\mathop{=}^{(a)}\left[{\rm diag}\left(\frac{\partial J}{\partial {\bf\widetilde \Theta}^*}{\bigg|}_{{\bf \widetilde \Theta}^*={\bf \Theta}^*}\right)\right]^T,
\end{align}
\noindent where $(a)$ holds since ${\bf L}_{\rm D}^T{\rm vec}({\bf A})={\rm diag}({\bf A})$ \cite{hjorungnes2011}. From (\ref{eq:chainrule2}), we obtain (\ref{eq:gradient}).

\section{Proof of \emph{Proposition \ref{pro:Eucgradtheta}}}\label{app:Eucgradtheta}

From \emph{Lemma \ref{lem:gradient}}, we obtain the derivatives of the terms in (\ref{Jderive}) with respect to ${\bm \theta}^*$ as follows
\begin{align}
\frac{\partial\left({\rm tr}\left({\bf Y}\circ{\bf Y}\right)\right)}{\partial{\bm \theta}^*}
&\mathop{=}^{(a)}{\rm diag}\left(\frac{2{\rm tr}\left(\left({\bf I}\circ{\bf Y}^T\right){\bf \Gamma}{\bf\widetilde A}{\bf \Theta}{\bf V}d\left({\bf \Theta}^H\right)
{\bf\widetilde A}^H\right)}{d{\bm \Theta}^*}\right)\nonumber\\
&={\rm diag}\left(2{\bf\widetilde A}^H\left({\bf I}\circ{\bf Y}^T\right){\bf \Gamma}{\bf\widetilde A}{\bf \Theta}{\bf V}\right)\label{eq:YY},\\
\frac{\partial\left({\rm tr}\left({\bf F}{\bf Y}^H\right)\right)}{\partial{\bm \theta}^*}
&={\rm diag}\left(\frac{{\rm tr}\left({\bf F}{\bf\widetilde A}{\bf \Theta}{\bf V}d\left({\bf \Theta}^H\right){\bf\widetilde A}^H{\bf \Gamma}\right)}{d{\bf \Theta}^*}\right)
={\rm diag}\left({\bf\widetilde A}^H{\bf \Gamma}{\bf F}{\bf\widetilde A}{\bf \Theta}{\bf V}\right),\label{eq:FY}
\end{align}
\noindent where the identity (\ref{eq:identity}) is used in $(a)$.

From (\ref{eq:YY})-(\ref{eq:FY}) and (\ref{Jderive}), we derive the gradient of $J$ as follows
\begin{align}\label{eq:Jderive2}
\frac{\partial J}{\partial{\bm \theta}^*}
&={\rm diag}\left({\bf\widetilde A}^H\left[2c^2\left({\bf I}\circ{\bf Y}^T\right){\bf \Gamma}-2c{\bf \Gamma}{\bf F}\right]{\bf\widetilde A}{\bf \Theta}{\bf V}\right),
\end{align}
from which we obtain (\ref{eq:Eucgradtheta}).

\section{Proof of \emph{Proposition \ref{pro:shiftAOA}}}\label{app:shiftAOA}

By synthesizing the power pattern based on an LoS BS-RIS channel path with the AoA of $\phi_0$, we obtain the reflected power pattern $y_{\rm L}(\phi)$ towards the AoD of $\phi$ as follows \cite{He2021}
\begin{align}\label{eq:f0}
y_{\rm L}(\phi)=N_{\rm BS}M^2\left|{\bf a}_{\bf H}^H(\phi){\bf \Theta}{\bf a}_{\bf G}(\phi_0)\right|^2.
\end{align}
Then, we consider another incident BS-RIS channel path with the AoA of $\phi_1\ (\phi_1\neq\phi_0)$, and rewrite the term ${\bf a}_{\bf H}^H(\phi){\bf \Theta}{\bf a}_{\bf G}(\phi_0)$ in (\ref{eq:f0}) as follows
\begin{align}\label{eq:athetaa}
{\bf a}_{\bf H}^H(\phi){\bf \Theta}{\bf a}_{\bf G}(\phi_0)
&=\frac{1}{M}\sum_{m=0}^{M-1}e^{-\jmath \frac{2\pi}{\lambda}m\rho \cos \phi} \theta_{m} e^{-\jmath \frac{2\pi}{\lambda}m\rho \cos \phi_0}\nonumber\\
&=\frac{1}{M}\sum_{m=0}^{M-1}e^{-\jmath \frac{2\pi}{\lambda}m\rho (\cos \phi+\xi)} \theta_{m} e^{-\jmath \frac{2\pi}{\lambda}m\rho \cos \phi_1},
\end{align}
\noindent where $\xi\triangleq\cos\phi_0-\cos\phi_1$. Observing (\ref{eq:athetaa}), we have $y_{\rm L}(\phi)=y_{\rm L}(\phi')$ if there exists $\phi'$ satisfying $\cos\phi'=\cos \phi+\xi$, which means that the flat-top beam at the angle of $\phi$ is shifted to the angle of $\phi'$.
Since cosine function is monotonous within $[0,\pi]$, the region of the shifted flat-top beam can be determined by solving the equations
$\cos\phi'_{\rm min}=\cos\phi_{\rm min}+\xi$ and $\cos\phi'_{\rm max}=\cos\phi_{\rm max}+\xi$ at the bounds of the flat-top beam $\mathcal A_1=[\phi_{\rm min},\phi_{\rm max}]$.
If $\cos\phi+\xi>1$ or $\cos\phi+\xi<-1$, the shifted flat-top beam is partly cut off by the range of  $[0,\pi]$.
In detail, when $\phi_1<\phi_0$, if $\cos\phi_{\rm max}+\xi<-1$, we obtain the cut-off bound $\phi'_{\rm max}=\pi$.
Therefore, $\phi'_{\rm max}=\arccos(\cos\phi_{\rm max}+\xi)$ or $\phi'_{\rm max}=\pi$ .
Moreover, when $\phi_1>\phi_0$, since we assume that $\phi_{\rm min},\phi_{\rm max}\in[\pi/2,\pi]$,  $\cos\phi_{\rm min}+\xi<1$ holds.
From the above discussions, we obtain the region of the shifted flat-top beam in (\ref{eq:A1shift}).

\section{Proof of \emph{Theorem \ref{thm:flattop}}}\label{app:lemflattop}

For a given user within the broad coverage, located at the angle of $\phi_{q}$, the average received power at the $i$-th antenna is
\begin{align}\label{eq:signal2}
{\mathbb E}\{|{\bf r}_i[k]|^2\}
&=p \beta_1 \beta_2 M{\mathbb E}\left\{\left\|\sum_{q=0}^{Q-1}\delta_q[k]e^{-\jmath \frac{2\pi}{\lambda}(i-1)\rho \sin \psi'_q}
{\bf a}_{\bf H}^H(\phi_{q})
{\bf \Theta}{\bf G}[k]{\bf W}[k]{\bf s}[k]\right\|_2^2\right\}+\sigma_z^2\nonumber\\
&\mathop{=}^{(a)}p \beta_1 \beta_2 M\left(\sum_{\phi_q\in\mathcal{A}_1}+\sum_{\phi_q\notin\mathcal{A}_1}\right){\mathbb E}\{|\delta_q[k]|^2\}
{\mathbb E}\left\{\left\|{\bf a}_{\bf H}^H(\phi_{q})
{\bf \Theta}{\bf G}[k]{\bf W}[k]{\bf s}[k]\right\|_2^2\right\}+\sigma_z^2\nonumber\\
&\mathop{=}^{(b)}p \beta_1 \beta_2 f_{\rm M}\sum_{\phi_q\in\mathcal{A}_1}{\mathbb E}\{|\alpha_q|^2\}+\sigma_z^2.
\end{align}
\noindent In ($a$), the independence of the channel path is applied. In ($b$), we substitute the power pattern sampled at the angle $\phi_{q}$, i.e., $y(\phi_{q})$ defined in (\ref{eq:avepattern}). 
Then, by using the assumption of the perfect pattern synthesis, i.e., $y(\phi)=f(\phi)$, the side lobe $f_{\rm S}=0$, and the rectangular roll-off $\mathcal{A}_3=\emptyset$, we observe that 
$y(\phi_{q})=f(\phi_{q})=0$ holds for $\phi_q\notin\mathcal{A}_1$, 
while $y(\phi_{q})=f(\phi_{q})=f_{\rm M}$ holds for $\phi_q\in\mathcal{A}_1$.
Without loss of generality, we assume that the K-factor of the RIS-UE channel is $K$ and the channel power of the NLoS paths is uniformly distributed over $[0,\pi]$.
In this way, only the LoS RIS-UE channel path and $|\mathcal{A}_1|/\pi$ of the NLoS RIS-UE channel paths are within the broad coverage, i.e., $\phi_q\in\mathcal{A}_1$.
Therefore, we replace the term $\sum_{\phi_q\in\mathcal{A}_1}{\mathbb E}\{|\alpha_q|^2\}$ in (\ref{eq:signal2}) by $K/(K+1)+|\mathcal{A}_1|/[\pi(K+1)]$, which results in a constant average power as given in (\ref{eq:UEpower2}).

\section{Proof of \emph{Proposition \ref{pro:fm}}}\label{app:fm}

The power of flat-top beam, i.e., $f_{\rm M}$, is the product of the incident power at the RIS and the array gain of the RIS.
Firstly, the overall power impinging the RIS increases linearly with the number of RIS reflecting elements, i.e., $M$.
On the other hand, as discussed in \cite{Kraus2008}, the gain ${\tilde f}$ of an array is defined as
\begin{align}\label{eq:gain}
{\tilde f}\triangleq \mu D=\frac{41253^{\Box}}{\theta^{\circ}_{\rm HP}\phi^{\circ}_{\rm HP}},
\end{align}
\noindent  where $\mu\ (0\leq\mu\leq 1)$ is the efficiency factor, $D$ is the directivity of the array, $41253\approx 4\pi(180/\pi)^2$ is the number of square degrees ($^{\Box}$) in sphere,  and $\theta^{\circ}_{\rm HP}$ and $\phi^{\circ}_{\rm HP}$ are the half-power beamwidths in two principal planes, i.e., the elevation and the azimuth, respectively. For a linear array, the angle of elevation $\theta^{\circ}_{\rm HP}$ is a constant. Therefore, the array gain of the RIS decreases linearly with the beamwidth of the power pattern defined in (\ref{eq:targetf}).

\section{Proof of \emph{Theorem \ref{thm:rateOFDMA}}}\label{app:thmOFDMA}

From (\ref{eq:rateMISO2}), we approximate the downlink rate ${\mathbb E}\{R'\}$ by using \cite[Lemma 1]{QIZHANG2014} as follows
\begin{align}\label{eq:rateMISOderive}
{\mathbb E}\{R'\}\approx& \sum_{u=0}^{U-1}\sum_{k\in\{{\mathcal K}_u\}}\log_2\left(1+\frac{p }{\sigma_z^2}{\mathbb E}\left\{\left\|\sqrt{\beta_1 \beta_2}{\bf h}_u^H[k]{\bf \Theta}{\bf G}_0[k]+\sqrt{\beta}{\bf h}_{{\rm d}u}^H[k]\right\|_2^2\right\}\right),
\end{align}
\noindent where the approximation becomes tighter as the number of BS antennas increases \cite{QIZHANG2014}.
Then, we derive the expectation of the terms in (\ref{eq:rateMISOderive}) separately.
Firstly, from (\ref{eq:hd}), we obtain
\begin{align}\label{eq:term2}
{\mathbb E}\left\{\left\|{\bf h}_{{\rm d}u}^H[k]\right\|_2^2\right\}=N_{\rm BS}\sum_{q'_u=0}^{Q'-1}{\mathbb E}\left\{|\alpha_{q'_u}|^2\right\}
=N_{\rm BS}.
\end{align}
Then, from (\ref{eq:hr}) and (\ref{eq:GL}), we obtain
\begin{align}
{\mathbb E}\left\{\left\|{\bf h}_u^H[k]{\bf \Theta}{\bf G}_0[k]\right\|_2^2\right\}
&=N_{\rm BS}M^2{\mathbb E}\left\{\left\|\sum_{q_u=0}^{Q-1}\delta_{q_u}[k]{\bf a}_{\bf H}^H(\phi_{q_u}){\bf \Theta}
\eta{\bf a}_{\bf G}(\phi_0){\bf b}_{\bf G}^H(\psi_0)\right\|_2^2\right\}\nonumber\\
&\mathop{=}^{(a)}\sum_{q_u=0}^{Q-1}{\mathbb E}\left\{|\alpha_{q_u}|^2\right\}
y_{\rm L}(\phi_{q_u})\nonumber\\
&\mathop{=}^{(b)}f_{\rm M}\sum_{\phi_{q_u}\in \mathcal{A}_1}{\mathbb E}\left\{|\alpha_{q_u}|^2\right\}\nonumber\\
&=f_{\rm M}\left[\frac{K}{K+1}+\frac{|\mathcal{A}_1|}{\pi(K+1)}\right]\label{eq:term1_2},
\end{align}
\noindent where (\ref{eq:f0}) is substituted in $(a)$, $\mathcal{A}_1$ is the flat-top region of the power pattern.  
In ($b$), the assumption of the perfect pattern synthesis is used as in \emph{Theorem \ref{thm:flattop}}, i.e., $y(\phi)=f(\phi)$, the side lobe $f_{\rm S}=0$, and the rectangular roll-off $\mathcal{A}_3=\emptyset$. 
Furthermore, by assuming that the power of NLoS RIS-UE channel paths is uniformly distributed over $[0,\pi]$, we can replace the term $\sum_{\phi_{q_u}\in\mathcal{A}_1}{\mathbb E}\{|\alpha_{q_u}|^2\}$ in (\ref{eq:term1_2}) by $K/(K+1)+|\mathcal{A}_1|/[\pi(K+1)]$ as previously derived in (\ref{eq:signal2}).

By substituting (\ref{eq:term2}) and (\ref{eq:term1_2}) into (\ref{eq:rateMISOderive}), we find that the analytical rate has the same expression at each subcarrier, therefore, we replace $\sum_{u=0}^{U-1}\sum_{k\in\{{\mathcal K}_u\}}$ in (\ref{eq:rateMISOderive}) by $N_{\rm c}$ and obtain (\ref{eq:rateOFDMA}).

\bibliographystyle{IEEEtran}

\bibliography{reference}

\end{document}